\documentclass[12pt,a4paper,fleqn]{article}

\usepackage[margin=1in]{geometry}
\usepackage[T1]{fontenc}
\usepackage{amsmath,dsfont}
\usepackage{caption}

\usepackage[utf8]{inputenc}
\usepackage{bm}
\usepackage{verbatim}
\usepackage{float}
\usepackage{mwe}
\usepackage{color,soul}
\usepackage{threeparttable}
\usepackage{graphicx}
\usepackage{subcaption}
\usepackage{mwe}
\usepackage{mdframed}
\usepackage{xcolor}
\usepackage{adjustbox}

\usepackage{natbib}
\usepackage{multirow}
\usepackage{makecell}
\usepackage{booktabs}
\usepackage{array}
\usepackage{hhline}

\usepackage{lscape}
\usepackage{mathtools}
\usepackage{bbm}
\usepackage{romanbar}
\usepackage{amsthm}
\usepackage{enumerate}
\usepackage{rotating}
\usepackage{tikz}
\usetikzlibrary{matrix}
\usetikzlibrary{calc,intersections}
\usepackage{mathrsfs}

\usepackage{color}

\newtheorem{theorem}{Theorem}[section]

\newtheorem{lemma}[theorem]{Lemma}

\theoremstyle{definition}

\theoremstyle{remark}
\newtheorem{remark}[theorem]{Remark}
\numberwithin{equation}{section}

\newcommand{\bol}[1]{\mbox{\boldmath$#1$}}

\newcommand{\bSigma}{\mathbf{\Sigma}}

\newcommand{\bmu}{\bol{\mu}}
\newcommand{\hbmu}{\hat{\bol{\mu}}}

\newcommand{\btheta}{\bol{\theta}}
\newcommand{\bb}{\mathbf{b}}

\newcommand{\bx}{\mathbf{x}}

\newcommand{\by}{\mathbf{y}}

\newcommand{\bC}{\mathbf{C}}

\newcommand{\bX}{\mathbf{X}}
\newcommand{\bY}{\mathbf{Y}}

\newcommand{\bw}{\mathbf{w}}
\newcommand{\hbw}{\mathbf{\hat{w}}}

\newcommand{\bOne}{\mathbf{1}}

\newcommand{\bI}{\mathbf{I}}

\newcommand{\bxi}{\boldsymbol{\xi}}
\newcommand{\bD}{\mathbf{D}}

\newcommand{\bV}{\mathbf{V}}

\newcommand{\bS}{\mathbf{S}}

\providecommand{\keywords}[1]
{
\small	
\textbf{\textit{Keywords}:} #1
}

\usepackage{xr}
\usepackage{authblk}
\title{Dynamic Shrinkage Estimation of the High-Dimensional Minimum-Variance Portfolio}
\author[1]{Taras Bodnar}
\author[2]{Nestor Parolya\footnote{Corresponding author: Nestor Parolya. E-mail address: n.parolya@tudelft.nl.}}
\author[1]{Erik Thors{\'e}n}
\affil[1]{Department of Mathematics, Stockholm University, Roslagsv\"{a}gen 101, SE-10691 Stockholm, Sweden}
\affil[2]{Department of Applied Mathematics, Delft University of Technology, Mekelweg 4,
2628 CD Delft, The Netherlands}
\date{\today}

\begin{document}
\maketitle

\begin{abstract}
    In this paper, new results in random matrix theory are derived which allow us to construct a shrinkage estimator of the global minimum variance (GMV) portfolio when the shrinkage target is a random object. More specifically, the shrinkage target is determined as the holding portfolio estimated from previous data. The theoretical findings are applied to develop theory for dynamic estimation of the GMV portfolio, where the new estimator of its weights is shrunk to the holding portfolio at each time of reconstruction. Both cases with and without overlapping samples are considered in the paper. The non-overlapping samples corresponds to the case when different data of the asset returns are used to construct the traditional estimator of the GMV portfolio weights and to determine the target portfolio, while the overlapping case allows intersections between the samples. The theoretical results are derived under weak assumptions imposed on the data-generating process. No specific distribution is assumed for the asset returns except from the assumption of finite $4+\varepsilon$, $\varepsilon>0$, moments. Also, the population covariance matrix with unbounded spectrum can be considered. The performance of new trading strategies is investigated via an extensive simulation. Finally, the theoretical findings are implemented in an empirical illustration based on the returns on stocks included in the S\&P 500 index.
\end{abstract}

\keywords{Shrinkage estimator; high-dimensional covariance matrix; random matrix theory; minimum variance portfolio; parameter uncertainty; dynamic decision making}

\textbf{JEL classification:} G11, C13, C14, C55, C58, C65

\section{Introduction}\label{sec:intro}

The global minimum-variance (GMV) portfolio is the one of the most commonly used investment strategies by both practitioners and researchers in finance. This portfolio possesses the smallest variance among all optimal portfolios obtained as solutions of Markowitz's mean-variance optimization problem (cf., \citet{markowitz1952}). It solves the following problem   
\begin{equation}\label{opt_GMV}
    \bw^\top \bSigma \bw \to min
    \quad \text{with} \quad \bw^\top \bOne_p=1, 
\end{equation}
where $\bw$ denotes the vector of the portfolio weights which determines the structure of the investor portfolio, the symbol $\bOne_p$ stands for the $p$-dimensional vector of ones, and $\bSigma$ is the covariance matrix of the $p$-dimensional vector of asset returns $\by=(y_1,...,y_p)^\top$. 

The solution to the optimization problem \eqref{opt_GMV} is given by
\begin{equation}\label{w_GMV}
    \bw_{GMV}=\frac{\bSigma^{-1}\bOne_p}{\bOne_p^\top\bSigma^{-1}\bOne_p}.
\end{equation}
The weights of the GMV portfolio have several nice properties, which simplify its applicability in practice and, thus, make it a popular investment strategy. The weights of the GMV portfolio do not depend on the mean vector of the asset returns, which we will denote by $\bmu$ in the following. This is the only mean-variance optimal portfolio whose weights are independent of $\bmu$. Moreover, the GMV portfolio has a special location on the set of the mean-variance optimal portfolios, which is a parabola in the mean-variance space and is known as the efficient frontier  (cf., \citet{merton1972}). Its mean and variance determines the location of the vertex of this parabola (see, e.g, \citet{kan2008distribution}, \citet{Bodnar2009Econometricalanalysisofthesampleefficientfrontier}). 

The application of \eqref{w_GMV} requires knowledge of $\bSigma$. In practice, we do not know what or how it looks. The covariance matrix $\bSigma$ has to be estimated using historical data of the asset returns, before the GMV portfolio can be constructed. The quality of the estimator of $\bSigma$ has a large impact on the stochastic properties of the GMV portfolio. It leads to further uncertainty in the investor decision problem, known as the estimation uncertainty. The estimation uncertainty can have a great impact on the constructed portfolio. It can be larger than the one induced by the model uncertainty included in the optimization problem \eqref{opt_GMV}. The effect becomes even stronger, when the portfolio dimension is comparable to the sample size used to estimate $\bSigma$.

Traditionally, the covariance matrix is estimated by its sample counterpart given by
\begin{equation}\label{bS}
\bS_n=\frac{1}{n-1}\sum_{i=1}^n (\by_i-\bar{\by}_n)(\by_i-\bar{\by}_n)^\top
=\frac{1}{n-1}\bY_n\left(\bI_n-\frac{1}{n}\bOne_n\bOne_n^\top\right)\bY_n^\top
\quad \text{with} \quad \bar{\by}_n=\frac{1}{n}\sum_{i=1}^n \by_i,
\end{equation}
where $\by_1,..,\by_n$ denotes the sample of asset returns and $\bY_n=(\by_1,..,\by_n)$ denotes the data matrix. The symbol $\bI_n$ stands for the $n$-dimensional identity matrix. Then, the traditional estimator, or empirical version, of $\bw_{GMV}$ is obtained as
\begin{equation}\label{hw_S}
    \hbw_{S}=\frac{\bS_n^{-1}\bOne_p}{\bOne_p^\top\bS_n^{-1}\bOne_p}.
\end{equation}
The distributional properties of $\hbw_{S}$ have extensively been studied in statistical and econometric literature. \citet{jobson1980estimation} derive the asymptotic distribution of $\hbw_{S}$ assuming that the asset returns are independent and normally distributed and the portfolio dimension is considerably smaller than the sample size. \citet{okhrin2006distributional} obtain the exact distribution of the sample estimator of the GMV portfolio weights assuming normality, while \citet{bodnar2008test} extend these results to elliptically contoured distribution and develop a statistical test theory on the GMV portfolio weights. 

Large portfolios are important to investors due to the diversification principle. However, when the portfolio dimension is comparable to the sample size many classical results does not hold. There is so much uncertainty such that we can not even guarantee consistency of the estimators. The results derived under the classical asymptotic regime, that is when $p$ considerably smaller than $n$, can no longer be used. The effect of dimensionality dominates the estimation of the covariance matrix. Using recent results from random matrix theory, several improved estimators for the weights of the GMV portfolio has been suggested when the portfolio dimension is comparable to the sample size. This is often called the large-dimensional asymptotic regime (see, e.g., \citet{bai2010spectral}). The properties of high-dimensional optimal portfolio weights are also studied by \citet{fan2012vast}, \citet{hautsch2015high}, \citet{ao2019approaching}, \citet{kan2019sample}, \citet{bodnar2020sampling}, \citet{cai2020high}, \citet{ding2020high}, among others.

One of the most commonly used methods to construct an improved estimator for the weights of the GMV portfolio is shrinkage estimators. These were first proposed by \citet{stein1956} with the aim to reduce the estimation error present in the sample mean vector computed for a sample from a multivariate normal distribution. Recently, this procedure has also been applied in the construction of the improved estimators of the high-dimensional mean vector (cf, \cite{chetelat2012}, \cite{wang2014}, \citet{bodnar2019optimal}), covariance matrix (see, e.g., \citet{lw2004}, \citet{lw12}, \citet{BodnarGuptaParolya2014}), inverse of the covariance matrix (see, e.g., \citet{wang2015shrinkage}, \citet{BodnarGuptaParolya2016}), as well as of the optimal portfolio weights (see, \citet{golosnoy2007multivariate}, \citet{frahm2010}, \citet{lw17}, \citet{bodnar2018estimation}, \citet{bodnarokhrinparolya2020}). Interval shrinkage estimators of optimal portfolio weights have recently been derived by \citet{BodnarDmytriv2019}, \citet{bodnar2020tests}. 

The shrinkage estimator for the weights of the GMV portfolio are obtained as a linear combination of the sample estimator $\hbw_{S}$ and the target portfolio $\bb$ with $\bb^\top \bOne=1$. The estimator is expressed as (see, \citet{bodnar2018estimation})
\begin{equation}\label{hw_SH}
\hbw_{SH}=\hat{\psi}_n\hbw_{S} + (1-\hat{\psi}_n)\bb \end{equation}
where
\begin{equation}\label{hpsi}
\hat{\psi}_n= \frac{(1-c_n) \hat{R}_{\bb}}{c_n+(1-c_n) \hat{R}_{\bb}} \quad \text{with}\quad \hat{R}_{\bb}=\left(1-c_n\right)
\bb^\top \bS_n\bb \bOne_p^\top\bS_{n}^{-1}\bOne_p-1
\quad \text{and}\quad 
c_n=\frac{p}{n}.
\end{equation}
\citet{bodnar2018estimation} show that this shrinkage estimator outperforms the sample estimator of the GMV portfolio weights in terms of minimizing the out-of-sample portfolio variance. The difference becomes drastic when $p$ approaches $n$. Moreover, the shrinkage estimator of the GMV portfolio weights \eqref{hw_SH} provides a simple and promising procedure on how one-period portfolio choice problems based on minimizing the portfolio variance can be solved in practice.

Once an optimal portfolio is determined, an investor face the problem of reallocation in the next period of time. One of the important decision to be made by the investor is to decide whether the holding portfolio is optimal or if it has to be adjusted (see, e.g., \citet{bodnar2009sequential}). As more data comes in, the estimators will change and in turn, the weights of the optimal portfolio.  
In the current paper we contribute to the literature by developing a dynamic GMV portfolio based on the shrinkage approach. At each time point of portfolio reconstruction the traditional estimator of the GMV portfolio weights is shrunk towards the weights of the holding portfolio, which by construction are the shrinkage estimator of the GMV portfolio from the previous period. The practical advantage of such a dynamic trading strategy are many. Making large changes can be costly and even be prohibited. Large institutions can not close large positions from one day to another since their movement might influence the market. Furthermore, depending on your loss it might not be optimal to transition from one portfolio to another without taking the concept of a transition into account. Especially when it comes to the out-of-sample variance as a loss function. 

From the perspectives of statistical theory, we develop new results that allow us to use the shrinkage estimators with a random target. These estimators are obtained under weak conditions imposed on the data-generating process. In particular, we only assume the existence of fourth moments together with a location and scale family. No other assumptions are made for the asset returns. Moreover, no assumption about the spectrum of the population covariance matrix is imposed in this paper. The eigenvalues of the covariance matrix can be as large as in factor models (see, e.g., \citet{fan2008high}, \citet{fan2013large}, \citet{ding2020high}). We only require that the ratio of the variances of the target portfolio and the GMV portfolio is bounded. Moreover, the derived theoretical results allow for application of overlapping samples in the determination of the target portfolio and construction of the traditional portfolio used in the specification of the shrinkage GMV portfolio.

The rest of the paper is organized as follows. In Section \ref{sec:seq_shrinkage}, the main theoretical findings of the paper are provided. The dynamic shrinkage estimator for the weights of the GMV portfolio is derived in the case of non-overlapping samples in Section \ref{sec:GMV_non-overlapping}, while Section \ref{sec:GMV_overlapping} presents the results in the overlapping case. The performance of the new trading strategies is investigated in Section \ref{sec:num_study} via an extensive simulation study, where the approaches are also compared to the existing ones. In Section \ref{sec:empirical}, the new approaches to estimate the GMV portfolio are implemented to the real data consisting of the returns on stocks included in S\&P 500 index. Concluding remarks are given in Section \ref{sec:summary}, while the technical proofs are moved to the appendix (see, supplementary material).

\section{Dynamic estimation of GMV portfolio}\label{sec:seq_shrinkage}

Throughout this paper we assume that the GMV portfolio is constructed at time point $t_1$ by using a sample of size $n_1$. The investor then updates the constructed GMV portfolio as new information arrives from the capital market. The information set is presented as a sequence of asset returns taken between time point $t_{i-1}$ and $t_{i}$ for $i=2,...,T$. Between each pairs $(t_{i-1},t_i)$ it is assumed that $n_i$ vectors of asset returns are available which are collected into the data matrix $\bY_{n_i}$ that is assumed to possess the following stochastic representation:
\begin{equation}\label{obs_i}
    \bY_{n_i}=\bmu \bOne_{n_i}^\top+\bSigma^{\frac{1}{2}}\bX_{n_i},
\end{equation}
where $\bX_{n_i}$ is a $p\times n_i$ matrix which consists of independent and identically distributed (i.i.d.) real random variables with zero mean and unit variance. We assume that the entries of $\bX_{n_i}$, $i=1,...,T$, possess the $4+\varepsilon$, $\epsilon>0$, moments, while no specific distributional assumption is imposed on the element of $\bX_{n_i}$. To this end, it is assumed that $\bY_{n_i}$, $i=1,...,T$, are independent random matrices. 


We consider an investor who aims to perform shrinkage estimation of the GMV portfolio weights in each period of time $t_i$. Namely, after constructing the shrinkage estimator of the GMV portfolio as defined in \eqref{hw_SH} at time point $t_1$, the investor updates the GMV portfolio weights by shrinking their sample estimator computed at each time point $t_i$ to the holding GMV portfolio determined at time point $t_{i-1}$. Two estimation strategies are developed in this section, which are based on non-overlapping and overlapping samples, respectively. The first procedure can be related to the rolling window estimation but with potentially different sample sizes. The main advantage here is that smaller sample sizes are used in the construction of the sample weights of the GMV portfolio and, thus, extreme observations or deviations in the asset returns can sooner be detected. Such a strategy might be recommendable during the turbulent period on the capital market, since it allows a faster adjustment of the holding portfolio. In contrary, when the stable period on the capital market is present, then the investor would prefer to use all available information. It leads to an extended window estimation strategy. In this case part of data used in the construction of the GMV portfolio has already been used to determine the currently holding portfolio. The new estimator will then contain new and old information and, consequently, we have a case with overlapping samples. Both situations require completely different techniques from random matrix theory to be developed in order to derive the stochastic properties of the estimation procedures, which are developed in the consequent two subsections.

\subsection{Dynamic GMV portfolio with non-overlapping samples}\label{sec:GMV_non-overlapping}

Under the non-overlapping scenario, the investor uses the sample of asset returns collected in $\bY_{n_i}$ to construct the sample estimator of the GMV portfolio at each time point $t_i$ expressed as
\begin{equation}\label{hw_Si}
    \hbw_{S;n_i}=\frac{\bS_{n_i}^{-1}\bOne_p}{\bOne_p^\top\bS_{n_i}^{-1}\bOne_p}
\quad \text{with} \quad
\bS_{n_i}=\frac{1}{n_i-1}\bY_{n_i}\left(\bI_{n_i}-\frac{1}{n_i}\bOne_{n_i}\bOne_{n_i}^\top\right)\bY_{n_i}^\top.
\end{equation}

At time point $t_i$ the shrinkage estimator of the GMV portfolio is then obtained by shrinking \eqref{hw_Si} to the weights of the holding portfolio, i.e., to the shrinkage estimator of the GMV portfolio $\hbw_{SH;n_{i-1}}$ constructed in the previous period. We do so by minimizing the loss function determined as the out-of-sample variance with respect to the shrinkage intensity $\psi_{n_i}$ in the following way:
\begin{equation}\label{L_i}
   \min_{\psi_{i}} L_i(\psi_{i})=\min_{\psi_{i}}\hbw_{SH;n_i}^\top \bSigma \hbw_{SH;n_i}
\end{equation}
with
\begin{equation}\label{w_SHi}
\hbw_{SH;n_{i}}=\psi_{i}\hbw_{S;n_i}+ (1-\psi_{i})\hbw_{SH;n_{i-1}},
\end{equation}
where $\hbw_{SH;n_{0}}=\bb$ is the shrinkage target used for the construction of the shrinkage estimator for the GMV portfolio weights at time point $t_1$.

Rewriting \eqref{L_i} we get
\[L_i(\psi_{i})=\psi_{i}^2\hbw_{S;n_i}^\top \bSigma \hbw_{S;n_i}+ 2\psi_{i}(1-\psi_{i})\hbw_{S;n_i}^\top\bSigma \hbw_{SH;n_{i-1}}+ (1-\psi_{i})^2\hbw_{SH;n_{i-1}}^\top\bSigma \hbw_{SH;n_{i-1}},\]
which is minimized at
\begin{equation}\label{alp_ni_star}
    \psi_{n_i}^*= \frac{\hbw_{SH;n_{i-1}}^\top \bSigma \left(\hbw_{SH;n_{i-1}}-\hbw_{S;n_i}\right)} {\left(\hbw_{SH;n_{i-1}}-\hbw_{S;n_i}\right)^\top \bSigma\left(\hbw_{SH;n_{i-1}}-\hbw_{S;n_i}\right)}.
\end{equation}

In Theorem \ref{th1} we derive the asymptotic equivalent to $\psi_{n_i}^*$ which can be used in the construction of the shrinkage estimator at time point $t_i$.

\begin{theorem}\label{th1}
Let $\bY_{n_i}$ possess the stochastic representation as in \eqref{obs_i} and let $\bb$ be the deterministic shrinkage target for $i=1$. Assume that the relative loss of portfolio $\bb$ given by 
\begin{equation}\label{relat_loss-bb}
r_{0}=\frac{V_{\bb}}{V_{GMV}}-1
=\bOne_p^\top \bSigma^{-1}\bOne_p\bb^\top\bSigma \bb-1    
\end{equation}
is uniformly bounded in $p$,
where 
\begin{equation}\label{var-bb-gmv}
V_{\bb}=\bb^\top \bSigma \bb \quad\text{and} \quad V_{GMV}=\bw_{GMV}^\top \bSigma \bw_{GMV}=\frac{1}{\bOne_p^\top \bSigma^{-1}\bOne_p}
\end{equation}
are the variances of the target portfolio $\bb$ and of the population GMV portfolio, respectively. Then it holds that
\begin{equation}\label{alp_i-star}
  \left|\psi_{n_i}^*-\psi_i^*\right| \stackrel{a.s.}{\rightarrow} 0 
  \quad \text{with} \quad
  \psi_i^*= \frac{(1-c_i)r_{i-1}}{(1-c_i)r_{i-1}+c_i}
\end{equation}
for $p/n_i \to c_i \in (0,1)$ as $n \to \infty$ where $r_i$ is the asymptotic equivalent of the relative loss $r_{\hbw_{SH;n_{i}}}= \bOne_p^\top \bSigma^{-1}\bOne_p \hbw_{SH;n_{i}}^\top \bSigma \hbw_{SH;n_{i}} -1$ of the portfolio with weights $\hbw_{SH;n_{i}}$ given by
\begin{equation}\label{r_i-star}
r_i= (\psi_{i}^*)^2\frac{c_i}{1-c_i}+(1-\psi_{i}^*)^2r_{i-1}
\end{equation}
for $i=1,...,T$.
\end{theorem}

The proof of Theorem \ref{th1} is given in the appendix (see, supplementary material). Its result provide a simple recursive algorithm how the shrinkage intensities have to be computed in practice. Independently of the number of portfolio reallocations, $T$, the only unknown quantity in the algorithm is the relative loss of of the target portfolio $\bb$ used in the construction of the shrinkage estimator at time $i=1$. Using the sample $\bY_{n_1}$ its consistent estimator is given by
\begin{equation}\label{hat_r0}
    \hat{r}_0=\left(1-\frac{p}{n_1}\right)\bOne_p^\top \bS_{n_1}^{-1}\bOne_p\bb^\top\bS_{n_1} \bb-1.
\end{equation}
Then, the resulting (bona fide) shrinkage estimator of the GMV portfolio at time $i$ is given by
\begin{equation}\label{w_BFi}
\hbw_{BF;n_{i}}=\hat{\psi}^*_{i}\hbw_{S;n_i}+ (1-\hat{\psi}^*_{i})\hbw_{BF;n_{i-1}}
\quad \text{with} \quad
\hat{\psi}^*_{i}= \frac{(n_i-p)r_{i-1}}{(n_i-p)r_{i-1}+p},
\end{equation}
where $\hat{r}_i$ is computed recursively by
\begin{equation}\label{hat_ri}
    \hat{r}_i=(\hat{\psi}_{i}^*)^2\frac{p}{n_i-p}
    +(1-\hat{\psi}_{i}^*)^2\hat{r}_{i-1}
\end{equation}
with $\hat{r}_0$ as in \eqref{hat_r0} and $\hbw_{BF;n_{0}}=\bb$.

We conclude this section with several important remarks:

\begin{remark}\label{rem:rem1}
The deterministic target portfolio $\bb$ can also be replaced by the sample GMV portfolio computed by using data available before the sample $\bY_{n_1}$ is taken. If we denote these data by $\bY_{n_0}$, then the target weights $\bb$ are replaced by
\begin{equation}\label{hw_S0}
    \hbw_{S;n_0}=\frac{\bS_{n_i}^{-1}\bOne_p}{\bOne_p^\top\bS_{n_0}^{-1}\bOne_p}
\quad \text{with} \quad
\bS_{n_0}=\frac{1}{n_0-1}\bY_{n_0}\left(\bI_{n_0}-\frac{1}{n_0}\bOne_{n_0}\bOne_{n_0}^\top\right)\bY_{n_0}^\top.
\end{equation}
In this case the relative loss $r_0$ does not longer depend on the population covariance matrix $\bSigma$ and following the proof of Theorem \ref{th1} it is given by
\[\tilde{r}_0=\frac{c_0}{1-c_0}\approx \frac{p}{n_0-p}.\]
As a result, the (bona fide) shrinkage estimator of the GMV portfolio weights is obtained as in \eqref{w_BFi} and \eqref{hat_ri} with $\hat{r}_0$ replaced by $\tilde{r}_0$ and $\hbw_{BF;n_{0}}=\hbw_{S;n_0}$. In a similar way other random targets can be employed into our model, e.g., nonlinear shrinkage \cite{lw12}, but then the asymptotics and estimation of $r_0$ becomes highly nontrivial. One needs to handle every one of those targets separately. Thus, because of the large number of possible competitors we leave this interesting topic on the choice of the vector $b$ for the future research and for the sake of brevity concentrate ourselves on the naive equally weighted target $b=\bOne_p/p$ in our simulation and empirical studies.
\end{remark}

\begin{remark}\label{rem:rem2}
The results of Theorem \ref{th1} are derived under very week conditions which require the existent of $4+\varepsilon$, $\varepsilon>0$, moments only. No structural assumption on $\bSigma$ neither on $\bb$ are imposed.
\end{remark}

\begin{remark}\label{rem:rem3}
Other consistent estimators for $r_0$ can also be constructed. For instance, we can update our estimator at each time point $t_i$ as soon as new data of the asset returns become available. Let $N_i=\sum_{j=1}^i n_j$ be the total number of asset return vectors available at time point $t_i$ and let $\bY_{N_i}$ be the $p\times N_i$ matrix of the asset returns up to time $N_i$ that is $\bY_{N_i}=(\bY_{n_1}\,\bY_{n_2}\,...\,\bY_{n_i}\,)$. Then, at time point $i$, a consistent estimator for $r_0$ is obtained by
\begin{equation}\label{hat_r0i}
    \hat{r}_{0;i}=\left(1-\frac{p}{N_i}\right)\bOne_p^\top \bS_{N_i}^{-1}\bOne_p\bb^\top\bS_{N_i} \bb-1.
\end{equation}
where $\bS_{N_i}$ is the sample covariance matrix based on the data matrix $\bY_{N_i}$ as given in \eqref{bS} with $n=N_i$. Then, the (bona fide) shrinkage estimator of the GMV portfolio weights is computed following \eqref{w_BFi} and \eqref{hat_ri} with $\hat{r}_0$ replaced by $\hat{r}_{0;i}$. Since larger dataset is used to estimate $r_0$, we expect that this approach will perform better as the one suggested in \eqref{hat_r0} - \eqref{hat_ri}. On the other side, we change our initial belief of $r_0$ so we should expect the weights to change. The new method is also harder to implement and potentially more time demanding, since the recursion in \eqref{hat_ri} has to be started from the beginning at each time $t_i$.
\end{remark}

\subsection{Dynamic GMV portfolio with overlapping samples}\label{sec:GMV_overlapping}
In this section, we present the shrinkage estimator for the GMV portfolio which is constructed based on overlapping samples. In Remark \ref{rem:rem3} it is suggested to use all available data $\bY_{n_1},\bY_{n_2},...,\bY_{n_i}$ up to time point $t_i$ to determine a consistent estimator for the relative loss $r_0$ of portfolio $\bb$. Here, we use a similar idea in the construction of the sample estimator of the GMV portfolio weights at time $t_i$. Such an approach possesses the advantage that we only require $n_1>p$, while the other sample sizes $n_2,...,n_T$ can be smaller than the portfolio dimension $p$.

Using the notations $N_i$, $\bY_{N_i}$, and $\bS_{N_i}$ introduced in Remark \ref{rem:rem3}, we define
\begin{equation}\label{hw_S_Ni}
    \hbw_{S;N_i}=\frac{\bS_{N_i}^{-1}\bOne_p}{\bOne_p^\top\bS_{N_i}^{-1}\bOne_p}.
\end{equation}
as the sample estimator of the GMV portfolio weights based on data of the asset returns included in $\bY_{N_i}$. Substituting $\hbw_{S;N_i}$ instead of $\hbw_{S;n_i}$ in \eqref{w_SHi}, the loss function $L_i(\psi_{i})$ in \eqref{L_i} is maximized at
\begin{equation}\label{alp_Ni_star}
    \Psi_{N_i}^*= \frac{\hbw_{SH;N_{i-1}}^\top \bSigma \left(\hbw_{SH;N_{i-1}}-\hbw_{S;N_i}\right)} {\left(\hbw_{SH;N_{i-1}}-\hbw_{S;N_i}\right)^\top \bSigma\left(\hbw_{SH;N_{i-1}}-\hbw_{S;N_i}\right)}.
\end{equation}

In Theorem \ref{th2} we derive an iterative procedure for computing the deterministic equivalents to $\Psi_{N_i}^*$ for $i=1,...,T$. The proof of Theorem \ref{th2} is given in the appendix (see, supplementary material).
\begin{theorem}\label{th2}
Let $\bY_{n_i}$ possess the stochastic representation as in \eqref{obs_i} and let $\bb$ be the deterministic shrinkage target for $i=1$. Assume that the relative loss of portfolio $\bb$ given by $R_{0}=\bOne_p^\top \bSigma^{-1}\bOne_p\bb^\top\bSigma \bb-1$ is uniformly bounded in $p$. Then it holds that
\begin{equation}\label{Psi_Ni}
  \left|\Psi_{N_i}^*-\Psi_i^*\right| \stackrel{a.s.}{\rightarrow} 0 
\quad \text{with} \quad
\Psi_i^*=\frac{(R_{i-1}+1)-K_i}
{(R_{i-1}+1)+(1-C_i)^{-1}-2K_i},
\end{equation}
for $p/N_j \rightarrow C_j \in (0, 1)$ as $N_j\rightarrow \infty$, $j=1,...,i$ and $i=1,...,T$ where
\begin{equation}\label{Psi_Ki}
K_i=\beta^{*}_{i-1;0}+\sum_{j=1}^{i-1} \beta^{*}_{i-1;j}D_{j,i}
\end{equation}
and
\begin{equation}\label{Psi_Ri}
R_i= (\Psi_{i}^*)^2\frac{C_i}{1-C_i}+(1-\Psi_{i}^*)^2R_{i-1}+2\Psi_{i}^*(1-\Psi_{i}^*)(K_i-1),
\end{equation}
with
\begin{equation}\label{Psi_betaji}
\beta^{*}_{0;0}=1, \quad \beta^{*}_{i-1;i-1}= \Psi_{i-1}^*,
\quad \text{and} \quad \beta^{*}_{i-1;j}= (1-\Psi_{i-1}^*)\beta^{*}_{i-2;j}, \, j=0,...i-2
\end{equation}
and
\begin{equation}\label{Psi_Dji}
D_{j,i}=1-\frac{2(1-C_j)}{(1-C_j)+(1-C_i)\frac{C_j}{C_i}+\sqrt{\left(1-\frac{C_j}{C_i}\right)^2+4(1-C_i)\frac{C_j}{C_i}}}.
\end{equation}
\end{theorem}

Similarly, to the case with non-overlapping samples, the recursive procedure derived in Theorem \ref{th2} depends only on a univariate unobservable quantity $R_{0}$, which is the relative loss of the target portfolio $\bb$ at time point $t_1$. Both approaches suggested in Section \ref{sec:GMV_non-overlapping} can be used to construct a consistent estimator for $R_{0}$, and hence to obtain a (bona fide) estimator of the GMV portfolio weights. These procedures are the following:
\begin{itemize}
    \item We estimate $R_0$ by
    \begin{equation}\label{hat_R0}
    \hat{R}_0=\hat{r}_0=\left(1-\frac{p}{N_1}\right)\bOne_p^\top \bS_{N_1}^{-1}\bOne_p\bb^\top\bS_{N_1} \bb-1
    \end{equation}
    as in \eqref{hat_r0}. In this case the estimator for $R_0$ is constructed by using the first sample $\bY_{N_1}$ only and the recursive procedure of Theorem \ref{th2} is then used leading to the (bona fide) optimal shrinkage estimators for the weights at each time point $t_i$, $i \in 1,...,T$ expressed as
    \begin{equation}\label{w_BFi-Psi}
\hbw_{BF;N_{i}}=\hat{\Psi}^*_{i}\hbw_{S;N_i}+ (1-\hat{\Psi}^*_{i})\hbw_{BF;N_{i-1}}
\end{equation}
where $\hat{\Psi}^*_{i}$ is computed recursively as in Theorem \ref{th2} with $R_0$ replaced by $\hat{R}_0$ and using the empirical counterpart for $C_i$ given by $C_{N_i}=p/N_i$.
    \item At each time point $i$, we use all available information to estimate $R_0$, i.e., 
    \begin{equation}\label{hat_R0i}
    \hat{R}_{0;i}=\hat{r}_{0;i}=\left(1-\frac{p}{N_i}\right)\bOne_p^\top \bS_{N_i}^{-1}\bOne_p\bb^\top\bS_{N_i} \bb-1
    \end{equation}
    and recompute the recursion of Theorem \ref{th2} at each time point $t_i$. Since a larger dataset is used to estimate $R_0$, better results are expected although the computations becomes more time demanding in the second case.
\end{itemize}

To this end, we note that the deterministic target portfolio $\bb$ can be replaced by the sample GMV portfolio computed by using data $\bY_{n_0}$ available before the first sample $\bY_{N_1}$ as in \eqref{hw_S0} of Remark \ref{rem:rem1}. In this case we get
\[\widetilde{R}_0=\frac{p}{n_0-p},\]
which is used in the iterative computation of Theorem \ref{th2} instead of $R_0$. Since no unknown quantities are present in the definition of $\widetilde{R}_0$, the iterative procedure of Theorem \ref{th2} becomes deterministic. 

\section{Finite-sample performance} \label{sec:num_study}
In this section we will study the performance of the new shrinkage methods in comparison to a number of benchmarks. These methods are implemented in the R-package DOSPortfolio (see e.g. \cite{DOSPortfolio}) available on CRAN.
\subsection{Benchmark strategies and the setup of the simulation study}\label{sec:benchmark}
In this section we present the benchmark strategies and the scenarios to be simulated from. The results of the empirical illustration are provided in Section \ref{sec:empirical}. We will investigate the performance of the following eight dynamic trading strategies: 
\begin{description}
    \item[Strategy 1:] Bona fide shrinkage estimator of the GMV portfolio \eqref{w_BFi} with \eqref{hat_ri} following Theorem \ref{th1} where $r_0$ is estimated from the first sample as in \eqref{hat_r0};
    \item[Strategy 2:] Bona fide shrinkage estimator of the GMV portfolio \eqref{w_BFi-Psi} where $\hat{\Psi}^*_{i}$ is computed recursively as in Theorem \ref{th2} and $R_0$ is estimated from the first sample as in \eqref{hat_R0};
    \item[Strategy 3:] Bona fide shrinkage estimator of the GMV portfolio \eqref{w_BFi} with \eqref{hat_ri} following Theorem \ref{th1} where $r_0$ is recomputed as in \eqref{hat_r0i} when a new sample becomes available;
    \item[Strategy 4:] Bona fide shrinkage estimator of the GMV portfolio \eqref{w_BFi-Psi} where $\hat{\Psi}^*_{i}$ is computed recursively as in Theorem \ref{th2} and $R_0$ is recomputed as in  \eqref{hat_R0i} when a new sample becomes available;
    \item[Strategy 5:] Sample estimator of the GMV portfolio computed at each time $t_i$, $i=1,2,...,T$, i.e., $\psi_i = 1$ in \eqref{w_SHi} for $i=1,2,...,T$;
    \item[Strategy 6:] Target portfolio $\bb$ used at each time $t_i$, $i=1,2,...,T$, i.e., $\psi_i = 0$ in \eqref{w_SHi} for $i=1,2,...,T$;
    \item[Strategy 7:] One-period shrinkage estimator of the GMV portfolio \eqref{hw_SH} with \eqref{hpsi} reconstructed at each time $t_i$, $i=1,2,...,T$ with equally-weighted portfolio as the target portfolio.
     \item[Strategy 8:] Ledoit-Wolf nonlinear shrinkage estimator of the GMV portfolio (see, \cite{lw17}) computed at each time $t_i$, $i=1,2,...,T$;
\end{description}

The first four strategies are based on the theoretical results derived in Sections \ref{sec:GMV_non-overlapping} and \ref{sec:GMV_overlapping} where two different methods for constructing bona fide shrinkage estimators of the GMV portfolio weights are explored following the discussion after Theorems \ref{th1} and \ref{th2}, respectively. \textbf{Strategies 5 to 7} are benchmark strategies which are based on the traditional estimator of the GMV portfolio, the target portfolio, and on the one-period shrinkage estimator. The last \textbf{Strategy 8} is the recent state-of-the-art method of \cite{lw17}, which efficiently applies the nonlinear shrinkage estimator of the covariance matrix on the GMV portfolio weights.

Since the GMV portfolio is the solution of the portfolio optimization problem with the aim to minimize the portfolio variance, the relative loss in the out-of-sample variance is used as a performance measure in this comparison study which for the portfolio with the estimated weights $\hat{\bw}$ is expressed as:
\begin{equation}\label{oos-rel-loss}
\text{Relative loss}\,(\bw)= \frac{\hat{\bw}^\top \bSigma \hat{\bw}-V_{GMV}}{V_{GMV}}= \bOne^\top \bSigma^{-1} \bOne\hat{\bw}^\top \bSigma \hat{\bw}-1,
\end{equation}
where we use the formula for the global minimum variance $V_{GMV}$ given in \eqref{var-bb-gmv}.

In the simulation study, we will look at two investment horizons $T=10$ and $T=20$. For each segment we let $n_i=250$, which would correspond to an investor who rebalances the holding portfolio on a yearly basis over 10 or 20 years. The parameters of the listed below models are simulated according to $\bmu=(\mu_1,...,\mu_p)^\top$ with $\mu_i \sim U(-0.2,0.2)$ and the covariance matrix $\bSigma$ is configured such that 20\% of the eigenvalues are equal to 0.2, 40\% equal to one and 40\% equal to 4, whereas the eigenvectors are generated from the Haar distribution. Following this simulation setup $\bSigma$ will have the same spectral distribution for all considered values of the concentration ratio $c_n$.

Four different stochastic models for the data-generating process will be considered, which are listed below:
\begin{description}
    \item[Scenario 1: $t$-distribution]
    The elements of $\bx_t$ are drawn independently from $t$-distribution with $5$ degrees of freedom, i.e., $x_{tj}\sim t(5)$ for $j=1,...,p$, while $\by_t$ is constructed according to \eqref{obs_i}.
    \item[Scenario 2: CAPM] 
    The vector of asset returns $\by_t$ is generated according to the CAPM (Capital Asset Pricing Model), i.e.,
    $$\by_t = \bmu + \bol{\beta} z_t+ \bSigma^{1/2}\bx_t,$$
    with independently distributed $z_t \sim N(0, 1)$ and $\bx_{t} \sim N_p(\mathbf{0}, \mathbf{I})$. The elements of vector $\bol{\beta}$ are drawn from the uniform distribution, that is $\beta_i \sim U(-1,1)$ for $i=1,...,p$.
    \item[Scenario 3: CCC-GARCH model of \cite{bollerslev1990modelling}] The asset returns are simulated according to 
    $$\by_t | \bSigma_t \sim N_p(\bmu, \bSigma_t)$$
    where the conditional covariance matrix is specified by 
    $$\bSigma_t = \bD_t^{1/2} \bC \bD_t^{1/2}
    \quad \text{with} \quad \bD_t = \operatorname{diag}(h_{1,t}, h_{2,t}, ..., h_{p,t}),$$
    where
    $$
    h_{j,t} = \alpha_{j,0} + \alpha_{j,1} (\by_{j, t-1} - \bmu_j)^2 + \beta_{j,1} h_{j, t-1}, \text{ for } j=1,2,...,p, \text{ and } t=1,2,...,n_i,~i=1,...,T.
    $$
    The coefficients of the CCC model are sampled according to $\alpha_{j,1} \sim U(0,0.1)$ and $\beta_{j,1} \sim U(0.6,0.7)$ which implies that the stationarity conditions, $\alpha_{j,1} + \beta_{j,1} < 1$, are always fulfilled. The intercept $\alpha_{j,0}$ is thereafter chosen such that the unconditional covariance matrix is equal to $\bSigma$.
    \item[Scenario 4: VARMA model] 
    The vector of asset returns $\by_t$ is simulated according to a 
    $$
    \by_{t} = \bmu + \mathbf{\Gamma} \by_{t-1} + \bSigma^{1/2}\bx_t
    \quad \text{with} \quad 
    \bx_t\sim N_p(\mathbf{0},\mathbf{I})
    $$
    for $t=1,...,n_i$, $i=1,...,T$, where $\mathbf{\Gamma} = \operatorname{diag}(\gamma_1, \gamma_2,..., \gamma_p)$ where $\gamma_i \sim U(-0.9, 0.9)$ for $i=1,...,p$.
\end{description} 

\textbf{Scenario 1} and \textbf{Scenario 2} fulfill the conditions imposed on the data-generating model in Section \ref{sec:seq_shrinkage}. The application of both scenarios result in samples that consist of independent random vectors with finite $4+\varepsilon$, $\varepsilon>0$, moments. Furthermore, the covariance matrix possesses finite eigenvalues in \textbf{Scenario 1}, while it has an unbounded spectrum in \textbf{Scenario 2} (cf.,  \citet{fan2013large}). On the other side, the samples obtained following \textbf{Scenario 3} and \textbf{Scenario 4} consist of dependent observations. In \textbf{Scenario 3} the random vector are uncorrelated although a non-linear dependence is present in the time series structure of the model, while the elements of the samples obtained from \textbf{Scenario 4} are strongly linearly dependent. 

For each segment of the time partition we generate a new sample of $n=250$ observations, which is applied in the computation of $\hbmu$, $\hat{\bS}_{n_i}$, $\hat{\bS}_{N_i}$, $\hbw_{S;n_i}$, and $\hbw_{SH; n_{i}}$. As a target portfolio we use the equally weighted portfolio with the weights $\bb= \bOne_p/p$. The results of the simulation study are based on 5000 independent runs from which the average relative loss is computed for each scenario, strategy and several values of the concentration ratio $c$. 

\subsection{Performance of the trading strategies}\label{sec:sim_res}
Figures \ref{fig:1} to \ref{fig:4} present the results of the simulation study for $i=3, i=6$, and $i=10$ when $T=10$ and for $i=6, i=13$ and $i=20$ when $T=20$. Interestingly, the computed average loss show a similar behaviour independently of the data-generating model used to draw the samples. This observation also holds in the case of \textbf{Scenario 3} and \textbf{Scenario 4}, which by construction do not fulfill the assumptions imposed on the data-generating model in the derivation of the theoretical results. As such, one can conclude that the presence of non-linear dependence structure between the observation vectors or even strong linear dependence has only minor impact on the validity of the results derived in Theorem \ref{th1} and Theorem \ref{th2}.

The best performance is obtained for \textbf{Strategy 2}, \textbf{Strategy 3}, and \textbf{Strategy 4}, which are followed by \textbf{Strategy 1}. The differences between the computed values for \textbf{Strategies 2 to 4} are very small and are present at the fourth decimal. All dynamic estimation strategies considerably outperform the four considered benchmark strategies, independently of the scenario used to generate samples. On the third place we rank the nonlinear shrinkage estimator, i.e., \textbf{Strategy 8}, while the single-period shrinkage estimator is ranked on the place four.
Finally, we note that for \textbf{Scenarios 1-4} the traditional estimator performs better than the portfolio strategy based on the target portfolio, when the concentration ratio $c$ is smaller than $0.75$, while it produces extremely large values of relative losses, when $c$ approaches one. This observation becomes even more prominent in case of the \textbf{Scenario 4}, where a strong autocorrelation was employed. Here already for $c=0.5$ the target portfolio starts outperforming the traditional estimator. To this end, we conclude that the dynamic re-estimation of the relative loss of the target portfolio $\bb$ shows a significant improvement when non-overlapping samples are used and the concentration ratio $c$ is relative large. In contrast, the application of the dynamic re-estimation of the relative loss in the case of overlapping samples leads to the considerably large computation time without large improvements. Finally, the increase of the trading horizon $T$ has only a minor impact on the plots presented in Figures \ref{fig:1} to \ref{fig:4}. The larger value of $T$ slightly reduces the computed average relative losses in the case of \textbf{Strategy 1}, while they become a slightly large for the single-period shrinkage approach.

\begin{figure}[h!]
\center
\includegraphics[width=\textwidth,keepaspectratio]{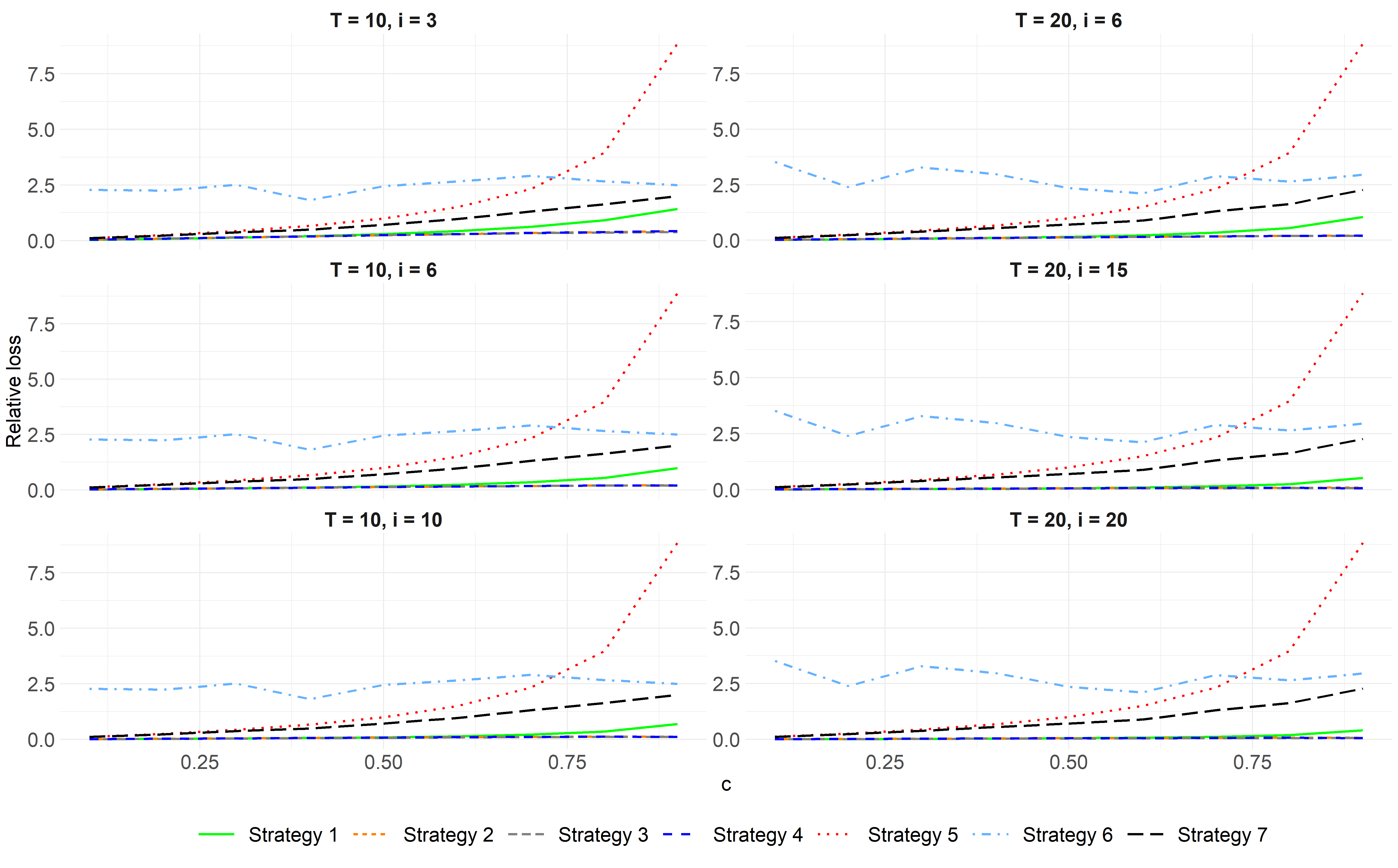}
\caption{Relative losses for the different time steps $i$ and investment horizons $T$. Data were simulated following \textbf{Scenario 1} for different values of $c$.}
\label{fig:1}
\end{figure}

\begin{figure}[h!]
\center
\includegraphics[width=\textwidth,height=\textheight,keepaspectratio]{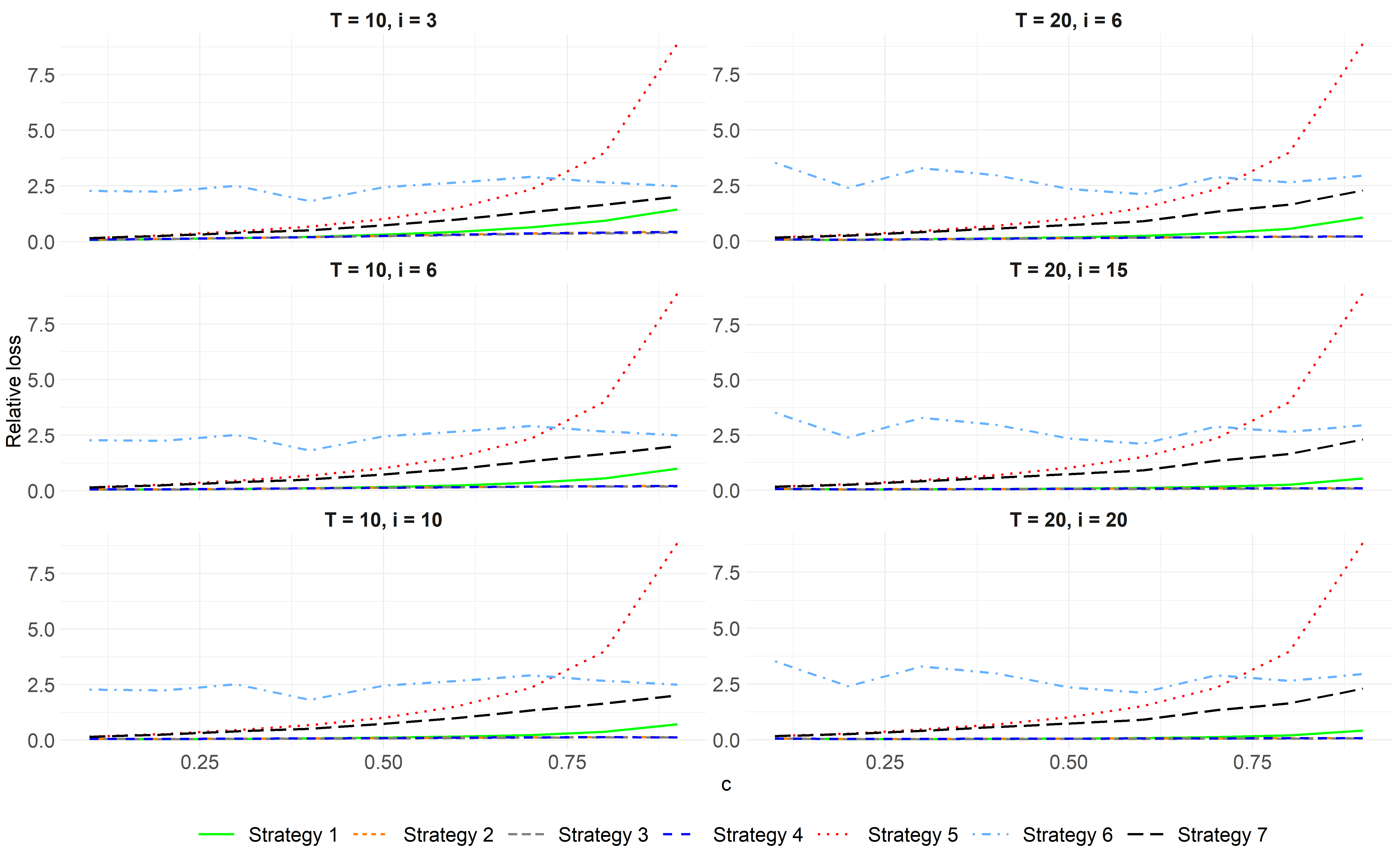}
\caption{Relative losses for the different time steps $i$ and investment horizons $T$. Data were simulated following \textbf{Scenario 2} for different values of $c$.}
\label{fig:2}
\end{figure}

\begin{figure}[h!]
\center
\includegraphics[width=\textwidth,height=\textheight,keepaspectratio]{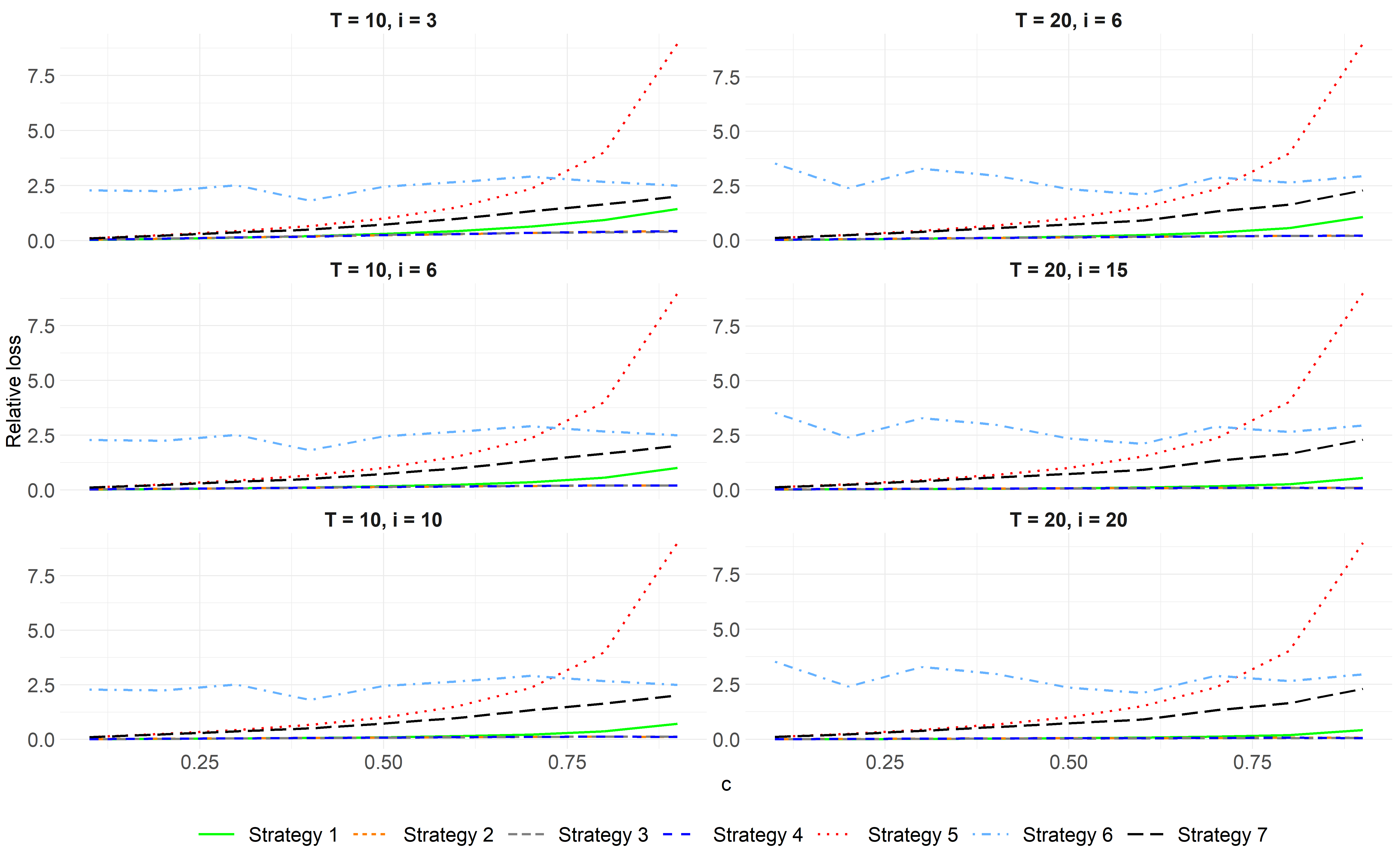}
\caption{Relative losses for the different time steps $i$ and investment horizons $T$. Data were simulated following \textbf{Scenario 3} for different values of $c$.}
\label{fig:3}
\end{figure}

\begin{figure}[h!]
\center
\includegraphics[width=\textwidth,height=\textheight,keepaspectratio]{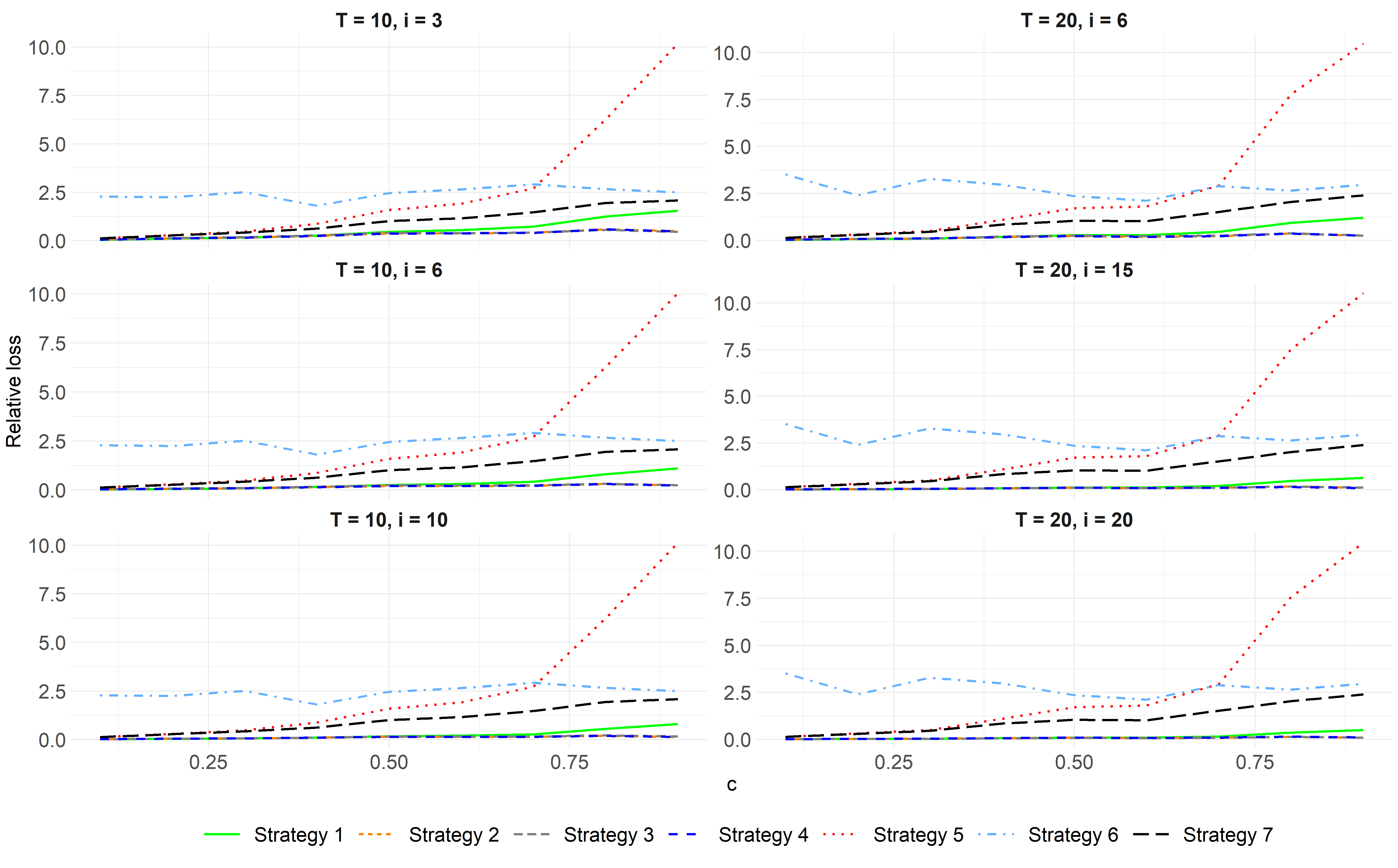}
\caption{Relative losses for the different time steps $i$ and investment horizons $T$. Data were simulated following \textbf{Scenario 4} for different values of $c$.}
\label{fig:4}
\end{figure}

\section{Application to stocks included in S\&P 500}\label{sec:empirical}

In this section we will apply the suggested new approaches and the benchmark strategies presented in Section \ref{sec:num_study} on daily market data.

\subsection{Data description}\label{sec:data}
We will use daily returns on $348$ stocks included in the S\&P500 index from March 2011 up January March 2021. The stocks were chosen by the availability of their price data during the trading period. Two portfolios of size $p=200$ and $p=150$ are considered, where $200$ stocks are chosen randomly from $348$ stocks included in the empirical study for the first portfolio, while the second one contains the first $150$ stocks in the alphabetic order of the former portfolio. We set $c=0.6$ or $c=0.8$ and will therefore use $n_i=250$ trading days for each year $i$. There will be $T=8$ reallocation points.

Figure \ref{fig:5} presents the descriptive statistics computed for the univariate series of the randomly chosen $200$ stocks. Namely, the first four (centralized) sample moments of the univariate time series of asset returns are depicted in the figure, which are computed for the whole period from March 2011 up until March 2021. The returns are on average positive over this period, while the sample skewness is on average negative with the sample distribution to be skewed to the left and with a few assets having very large (negative) values. The largest negative skewness corresponds to the Mondelez International, Inc. (ticker MDLZ) which is also among those assets with highest kurtosis. Finally, we note that the computed kurtosis are relatively large for the considered stocks, showing that the assumption of normality might not be fulfilled for the considered data.

\begin{figure}[h!]
\center
\includegraphics[width=\textwidth,height=\textheight,keepaspectratio]{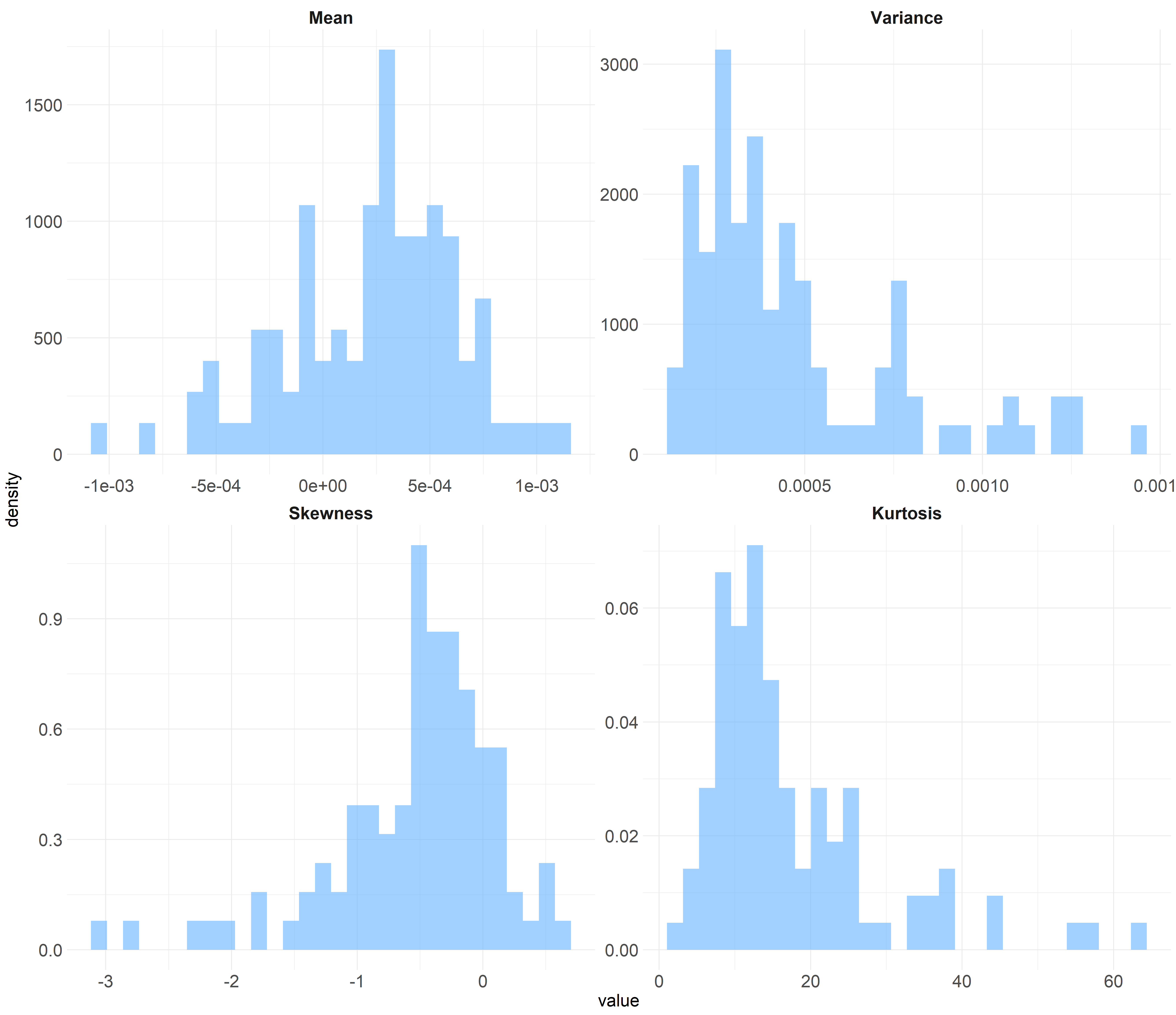}
\caption{\small First four moments for the univariate time series of the asset returns from 200 randomly selected assets of the S\&P500 index. The data consist of daily returns from March of 2011 to January of 2021.}
\label{fig:5}
\end{figure}

\subsection{Results of the empirical illustration}\label{sec:emp_results}
A consequence of the exponential weighting schemes to which the shrinkage estimators belong to, is that the portfolio structure changes by smaller increments. As a result, we expect that the portfolio turnover of the dynamic (estimated) GMV portfolios based on the introduced shrinkage approaches to be smaller in comparison to the unconstrained strategy ($\psi_i=1$), but to be larger in comparison to the static portfolio choice ($\psi_i=0$). For each strategy $k$ introduced in Section \ref{sec:benchmark}, let $\bw^{(k)}$ denote the vector of the weights induced by the $k$th strategy and let $w_{i,j}^{(k)}$ stand for the weight for the $j$-th asset after the $i$-th portfolio rebalancing. 

For each strategy $k$ the turnover is defined by (see e.g. \citet{golosnoy2019exponential})
\begin{equation}\label{turn}\centering
\text{Turnover}^{(k)}=\frac{1}{T}\sum_{i=1}^T ||\bw^{(k)}_i - \bw^{(k)}_{i-1}||_1.
\end{equation}
The turnover can be seen as the cost for transitioning from one portfolio to another, given that the transaction costs are constant for all assets and time periods. The amount of turnover will affect the development of wealth of the portfolio. Moreover, following \citet{golosnoy2019exponential}, we will compute the average absolute values of holding portfolio weights, the average minimum and maximum portfolio weights, the average sum of negative weights in the portfolio, and the average fraction of negative weights in the portfolio as further performance measures. They are given by

\vspace{-0.5cm}

\begin{align}\label{portmes1}
    |\bw^{(k)}| &= \frac{1}{Tp} \sum_{i=1}^T\sum_{j=1}^p |w_{i,j}^{(k)}|,\\\label{portmes2}
    \max \bw^{(k)} &= \frac{1}{T} \sum_{i=1}^T \left( \max_j w_{i,j}^{(k)}\right),\\\label{portmes3}
    \min \bw^{(k)} &= \frac{1}{T} \sum_{i=1}^T \left( \min_j w_{i,j}^{(k)}\right),
\\\label{portmes4}
\bw_i^{(k)} \mathbbm{1}(\bw_i^{(k)} < 0) &= \frac{1}{T} \sum_{i=1}^T\sum_{j=1}^p w_{i,j}^{(k)}\mathbbm{1}(w_{i,j}^{(k)} < 0), \\\label{portmes5}
    \mathbbm{1}(\bw_i^{(k)} < 0) &= \frac{1}{Tp} \sum_{i=1}^T\sum_{j=1}^p \mathbbm{1}(w_{i,j}^{(k)} < 0).
\end{align}
Moreover, we also consider important classical portfolio performance measures: total excess portfolio return, out-of-sample variance and average Sharpe ratio.
The computed values of the introduced performance measures are summarized in Table \ref{tab:1} for \textbf{Strategies 1 to 8} over the entire period. The corresponding values for \textbf{Strategy 6} are also included in the table but many of its entries are obviously equal to zero since the weights are all equal.  
In fact, \textbf{Strategy 6} has the smallest values of them all for the first 5 performance measures. Similarly, no negative weights are present in the target portfolio since each asset will be assigned a weight of $1/p$. Due to this fact we have chosen to highlight the second to best strategy, which is \textbf{Strategy 8}. If the investor is interested in the portfolio weights, then \textbf{Strategy 8} is the best as it takes the smallest positions, short the smallest amount of stocks and so forth. The average maximum weights, mean of all shorted and the proportion shorted seem to increase slightly with $p$. It seems as though when $p$ increases, \textbf{Strategy 8} decrease in the other metrics regarding the portfolio weights. The same is not true for the other strategies, especially \textbf{Strategy 5}. This can most likely be attributed to the nonlinear shrinkage of the eigenvalues, which seem to have a direct impact on the magnitude and direction of the weights. None of the other strategies has the same flexibility and can not compete with it.
The other strategies seem to have natural ordering in terms of the first five metrics. \textbf{Strategy 1} through \textbf{4} are better than \textbf{Strategy 7}, which in turn is better than \textbf{Strategy 5}.
One interesting feature of the re-estimation strategies, \textbf{Strategy 3} and \textbf{4}, does not imply a large change in the metrics. One of the largest changes can be seen in the mean out of the shorted weights for \textbf{Strategy 4} when $p=200$. When we re-estimate the initial relative loss, we change the shrinkage coefficients and all dynamic GMV portfolios thereafter. For the portfolio weights, that seem to imply an increase in certain metrics.

The last performance measures, seen in the four last rows of Table \ref{tab:1} are based on the portfolio return.  \textbf{Strategy 3} and \textbf{Strategy 4} have the largest portfolio return. Since we only present up to 4 decimals both are highlighted. In reality, they differ on the sixth decimal, which can of course make a big difference depending on the capital invested. In contrast, \textbf{Strategy 1} and \textbf{Strategy 2} beat all of other strategies in terms of portfolio variance when $p=150$ and \textbf{Strategy 2} has the least amount of variance when $p=200$. We can also see that as $p$ increases, the variances of \textbf{Strategies 5 to 8} increase. The dynamic shrinkage becomes an increasingly important feature in this context. It does not suffice to optimize for the out-of-sample variance, which \textbf{Strategies 1 to 4} and \textbf{Strategy 7} does, but also the dynamic nature of the investment process. These decrease with $p$. In turn, for $p=150$ \textbf{Strategy 4} gives the highest Sharpe Ratio and for $p=200$ \textbf{Strategy 2} gives the highest Sharpe Ratio. From the out-of-sample variance it is not surprising to see that \textbf{Strategy 5} provides the worst Sharpe ratio.  

\begin{landscape}
\begin{table}[ht]
\centering
\caption{Performance measures over $T=8$ periods of portfolio rebalancing. The strategy which is most performant (in terms of its measurement) is highlighted in bold on each row.Each Strategy is abbreviated "S.".} 
\label{tab:1}
\begin{adjustbox}{scale=1}
\begin{tabular}{|l|l|rrrrrrrr|}
  \hline
Type & Portfolio size & S. 1 & S. 2 & S. 3 & S. 4 & S. 5 & S. 6 & S. 7 & S. 8 \\ 
  \hline
  \multirow{2}{*}{$|\bw^{(k)}|$} 
    & 150 & 0.0245 & 0.0248 & 0.0255 & 0.0265 & 0.0477 & 0.0067 & 0.0298 & \textbf{0.0209} \\ 
    \cline{2-10}
    & 200 & 0.0218 & 0.0228 & 0.0224 & 0.0248 & 0.0627 & 0.0050 & 0.0300 & \textbf{0.0173} \\
  \hline
  \multirow{2}{*}{$\max \bw^{(k)}$} 
    & 150 & 0.1107 & 0.1120 & 0.1568 & 0.1633 & 0.2434 & 0.0067 & 0.1537 & \textbf{0.0836} \\ 
    \cline{2-10}
    & 200 & 0.1736 & 0.1812 & 0.1642 & 0.1823 & 0.5628 & 0.0050 & 0.2633 & \textbf{0.0865} \\
  \hline
  \multirow{2}{*}{$\min \bw^{(k)}$} 
    & 150 & -0.1155 & -0.1171 & -0.1018 & -0.1065 & -0.1882 & 0.0067 & -0.1134 & \textbf{-0.0652} \\ 
    \cline{2-10}
    & 200 & -0.0938 & -0.0983 & -0.0940 & -0.1052 & -0.3338 & 0.0050 & -0.1554 & \textbf{-0.0523} \\
  \hline
  \multirow{2}{*}{$\bw_i^{(k)} \mathbbm{1}(\bw_i^{(k)} < 0)$} 
    & 150 & -1.3387 & -1.3617 & -1.4102 & -1.4903 & -3.0742 & 0.0000 & -1.7380 & \textbf{-1.0648} \\
    \cline{2-10}
    & 200 & -1.6802 & -1.7767 & -1.7367 & -1.9807 & -5.7679 & 0.0000 & -2.4983 & \textbf{-1.2292} \\
  \hline
  \multirow{2}{*}{$\mathbbm{1}(\bw_i^{(k)} < 0)$}  
    & 150 & 0.4592 & 0.4592 & 0.4392 & 0.4442 & 0.4750 & 0.0000 & 0.4392 & \textbf{0.4050} \\
    \cline{2-10}
    & 200 & 0.4513 & 0.4519 & 0.4469 & 0.4656 & 0.4863 & 0.0000 & 0.4469 & \textbf{0.4250} \\ 
  \hline
  \multirow{2}{*}{$\bar{\by}_{\bw^{(k)}}$} 
    & 150 & 0.0003 & 0.0003 & \textbf{0.0004} & \textbf{0.0004} & 0.0001 & 0.0003 & 0.0002 & 0.0001 \\
    \cline{2-10}
    & 200 & 0.0003 & 0.0003 & \textbf{0.0004} & \textbf{0.0004} & -0.0001 & 0.0003 & 0.0001 & 0.0002 \\
  \hline
  \multirow{2}{*}{$\sigma_{\bw^{(k)}}$} 
    & 150 & \textbf{0.0062} & \textbf{0.0062} & 0.0083 & 0.0083 & 0.0123 & 0.0120 & 0.0112 & 0.0099 \\
    \cline{2-10}
    & 200 & 0.0055 & \textbf{0.0054} & 0.0080 & 0.0078 & 0.0156 & 0.0122 & 0.0124 & 0.0102 \\ 
  \hline
  \multirow{2}{*}{$\operatorname{SR}^{(k)}$} 
    & 150 & 0.0485 & 0.0486 & 0.0513 & \textbf{0.0525} & 0.0107 & 0.0250 & 0.0143 & 0.0138 \\
    \cline{2-10}
    & 200 & 0.0510 & \textbf{0.0514} & 0.0464 & 0.0496 & -0.0067 & 0.0208 & 0.0069 & 0.0180 \\
  \hline
  \multirow{2}{*}{$\text{Turnover}^{(k)}$} 
    & 150 & \textbf{0.0462} & 0.0507 & 1.4516 & 1.5016 & 8.6568 & 0.0000 & 5.2351 & 3.1340 \\ 
    \cline{2-10}
    & 200 & \textbf{0.1193} & 0.1372 & 1.9099 & 2.0772 & 15.5228 & 0.0000 & 7.1262 & 3.4561 \\
   \hline
\end{tabular}
\end{adjustbox}
\end{table}

\end{landscape}

The last performance measure in the last row is turnover.
As one can expect, \textbf{Strategy 5} is worst, generating the largest turnover. This classic strategy has the most flexibility. However, this flexibility often leads to unnecessary reconstruction of the holding portfolio. This becomes apparent in more extreme case when $p/n$ is close to one.
The smallest turnover is given by \textbf{Strategy 1}. The dynamic shrinkage force small movements between reallocations. However, in \textbf{Strategy 3} and \textbf{Strategy 4} we change the perception of the initial relative loss. This implies that algorithm change its opinion of what the optimal weights should have been. It causes the weights to change more relative to \textbf{Strategy 1} and \textbf{Strategy 2}. The re-estimation caused a larger return for $p=150$, a higher variance but also, in turn, a higher turnover. The dynamic shrinkage approach does not optimize towards decreasing turnover but it is a consequence of enforcing small movements. This also implies that other strategies can be close to the dynamic shrinkage approach. \textbf{Strategy 8} shows promise of this feature. When $p$ increases the nonlinear shrinkage estimator becomes more stable (between reallocation periods) and increases its turnover slightly. 

\begin{figure}[h!]
\center
\includegraphics[width=\textwidth,height=\textheight,keepaspectratio]{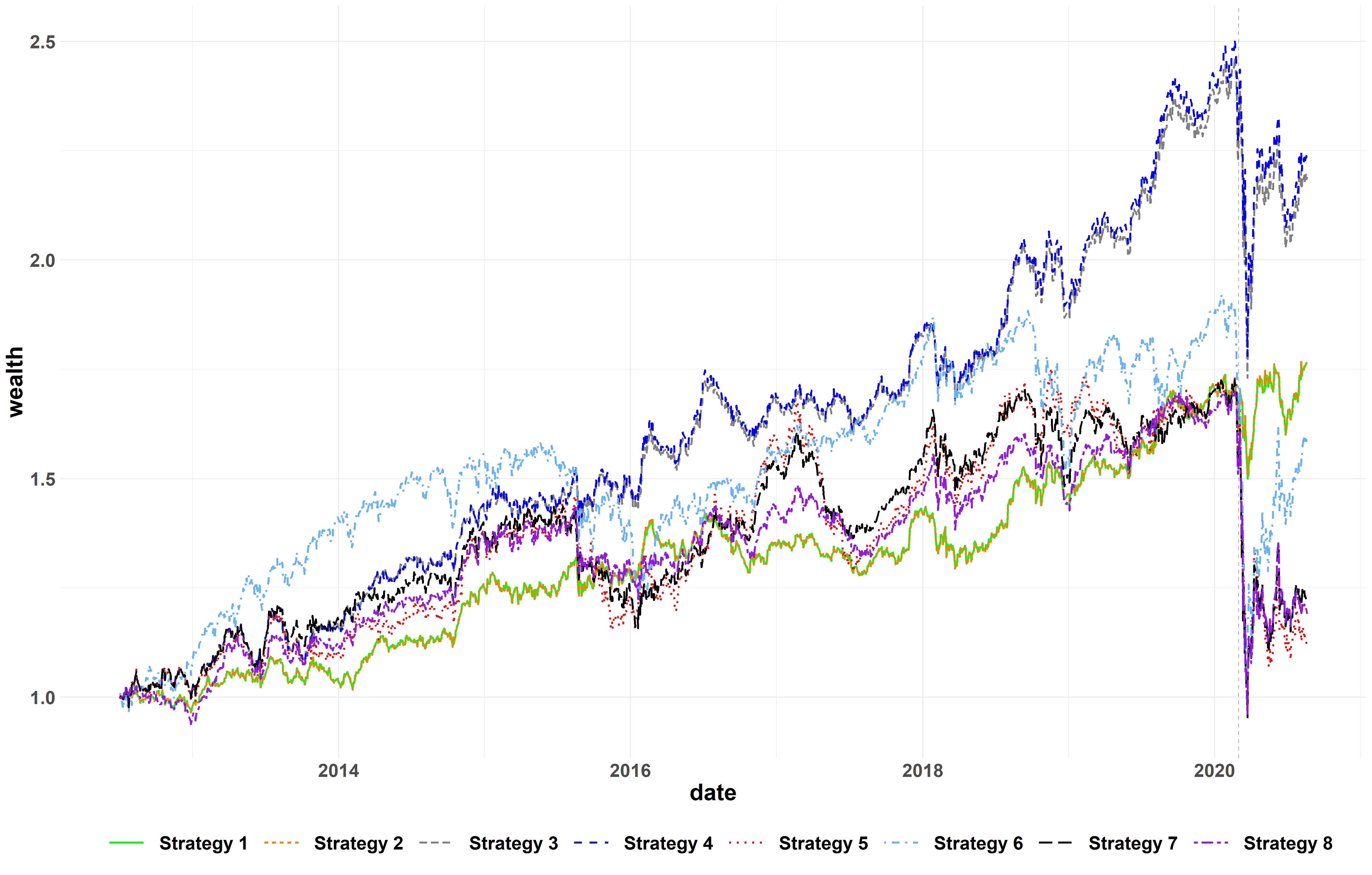}
\caption{Development of the investor wealth based on the dynamic trading strategies described in Section \ref{sec:benchmark}. In this figure the portfolio size is equal to 150.}
\label{fig:6}
\end{figure}

\begin{figure}[h!]
\center
\includegraphics[width=\textwidth,height=\textheight,keepaspectratio]{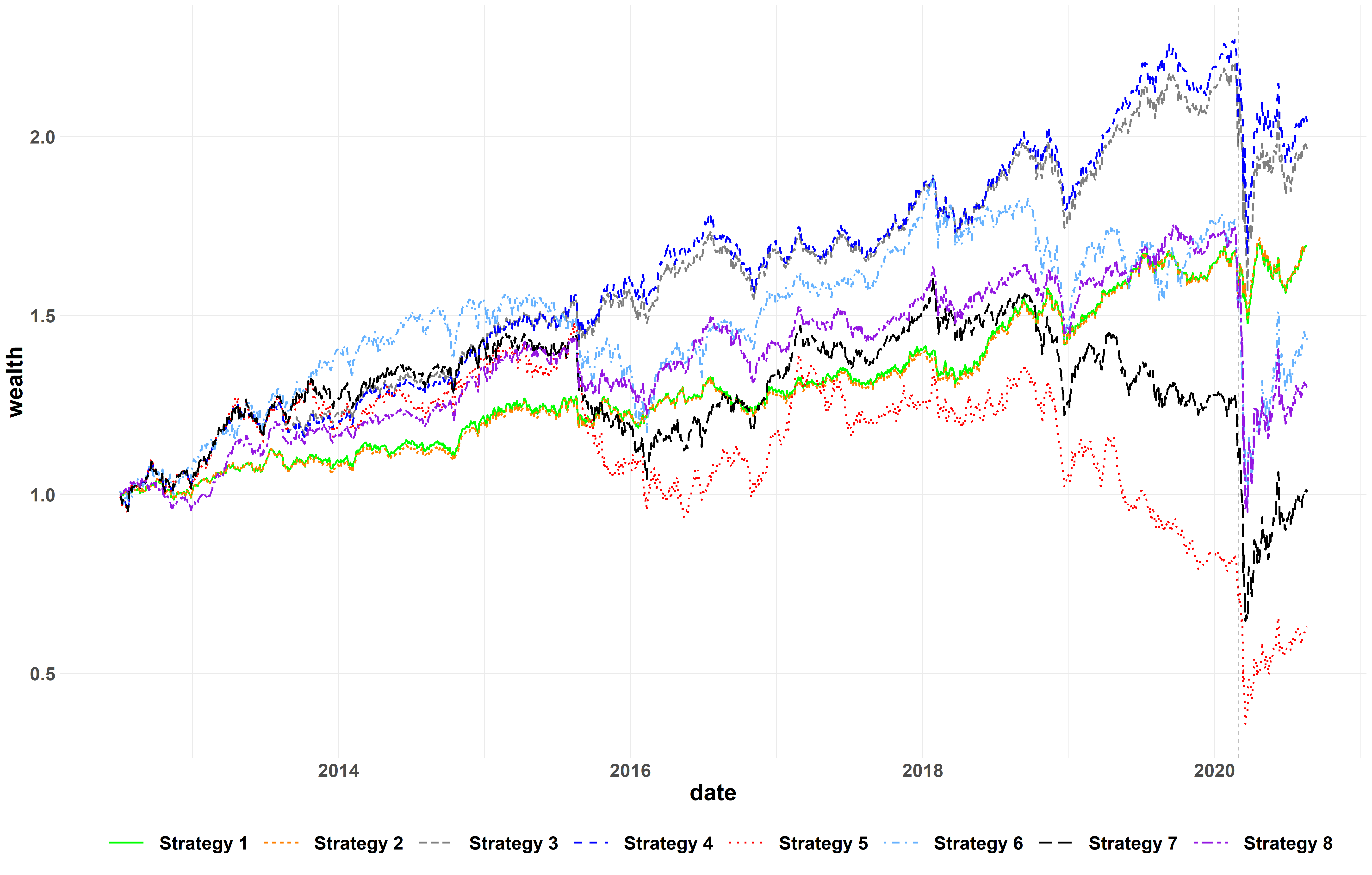}
\caption{Development of the investor wealth based on the dynamic trading strategies described in Section \ref{sec:benchmark}. In this figure the portfolio size is equal to 200.}
\label{fig:7}
\end{figure}


The development of the investors wealth for the eight trading strategies introduced in Section \ref{sec:benchmark} over 10 years is depicted in Figures \ref{fig:6} and \ref{fig:7}. The wealth is computed according to a buy-and-hold strategy until next reallocation period. That is, given that $n_i$ days has passed, we use these to estimate the portfolio weights and then rebalance the holding portfolio to the new portfolio according to the different strategies. The wealth is accumulated on a daily basis which corresponds to the frequency of data used to construct the portfolio.

The portfolios are the same as presented in \ref{tab:1}, e.g., one set of assets with $p=150$, $p/n_i=0.6$ and another with $p=200$, $p/n_i=0.8$.
In both scenarios the trading strategies based on the dynamic shrinkage approach using the re-estimation technique produce the largest wealth at the end of the investing period.
In Figure \ref{fig:6} where $p=150$ the worst is \textbf{Strategy 5} in terms of the wealth. However, there is a number of strategies that have close to the same performance. These are \textbf{Strategy 7} and \textbf{Strategy 8}. Throughout the whole period they are very close to each other where \textbf{Strategy 8} deviates the most from the other two. The next grouping is \textbf{Strategy 1}, \textbf{Strategy 2} and \textbf{Strategy 6}, the equally weighted portfolio. \textbf{Strategy 6} provides slightly less increase in wealth. The other two are very close in terms of their final wealth. They also seem to move in tandem. Their wealth is similar but from \ref{tab:1} we can see that \textbf{Strategy 1}, \textbf{Strategy 2} perform slightly different. The classical benchmark, \textbf{Strategy 6}, starts of accumulating a lot of wealth, but decrease just prior to 2016. In this figure, it also becomes clear why \textbf{Strategy 1} and \textbf{Strategy 2} had the smallest variance in Table \ref{tab:1}. They vary very little in contrast to other strategies.

In Figure \ref{fig:7} where $p=200$, we have a slightly different scenario. \textbf{Strategy 3} and \textbf{4} still provide the largest final wealth and \textbf{Strategy 1} and \textbf{2} second to that. Thereafter there seem to be a clearer separation between the rest in terms of final wealth. \textbf{Strategy 6} is still the best out of \textbf{Strategies 5 through 8}. However, \textbf{Strategy 8} is very close to deliver the same performance. The traditional GMV portfolio takes on both large turnover and does not provide any decent result. It almost cuts the wealth by half in the end of this period. The last strategy, \textbf{Strategy 7} is able to catch up and ends on a positive note.  

All portfolios are hit quite heavily by COVID in the early 2020, which is indicated by a dashed line on the 1st of March, 2020\footnote{This is of course somewhat arbitrary since it is hard to specify a certain day that COVID hits the market.}. Some portfolios are quick to adapt to the event and other are not. However, all portfolios seem to experience a very sharp increase in wealth post COVID. This can most likely be attributed to the very bullish scenario which was caused by banks all over the world pushing capital into the market. In Table \ref{tab:2} we show the largest negative difference of the strategies wealth throughout the whole period. We can see that \textbf{Strategy 2} has the smallest decrease in wealth between days during this period where \textbf{Strategy 1} comes in second. These strategies seem to have been robust against the COVID crisis. Thereafter, \textbf{Strategy 3} and \textbf{Strategy 4} have the smallest decrease. However, these are almost twice the size. The rest of the strategies are on a similar level. They all exhibit close to a drop of $0.2$ units in wealth.

\begin{table}[ht]
\centering
\caption{The largest difference (most negative) in wealth between days throughout the whole period. The strategy which is most performant (in terms of its measurement) is highlighted in bold on each row. Each Strategy is abbreviated "S.".} 
\label{tab:2}
\begin{tabular}{lrrrrrrrr}
  \hline
size & S. 1 & S. 2 & S. 3 & S. 4 & S. 5 & S. 6 & S. 7 & S. 8 \\ 
  \hline
150 & -0.0586 & \textbf{-0.0580} & -0.1180 & -0.1174 & -0.1975 & -0.2200 & -0.2048 & -0.1744 \\ 
  200 & -0.0469 & \textbf{-0.0442} & -0.1226 & -0.1190 & -0.2701 & -0.2266 & -0.2482 & -0.2009 \\ 
   \hline
\end{tabular}
\end{table}

These results are in line with the previous empirical findings of \citet{bodnar2020tests} who document that the equally weighted portfolio performs well in the stable period on the capital market, but its performance is very bad during the turbulent periods. To conclude, all four of the proposed dynamic shrinkage strategies show impressively good performance over the state-of-the-art static portfolios especially in case when $p$ becomes close to $n_i$.

\section{Summary}\label{sec:summary}
In many practical situations an investor faces the problem of portfolio reallocation. Especially when new data arrive from the market after the portfolio was built. We deal with this challenging task by developing several dynamic optimal shrinkage estimators for the weights of the GMV portfolio. In the derivation of the theoretical findings, new results in random matrix theory are deduced which allow us to obtained optimal shrinkage estimator in two important cases. Optimal portfolio estimators using with and without overlapping samples. In the case of non-overlapping samples, the investor uses the data of asset returns after the last reconstruction of the portfolio, while the whole data would be used in the case of overlapping samples. It is remarkable that the two settings require different theoretical results in random  matrix theory and they result in quite different optimal shrinkage intensities. Moreover, only minor distribution assumption are imposed on the data-generating process, like the existence of $4+\varepsilon$, $\varepsilon>0$ moments and the existence of a location and scale family. We want to emphasize that the covariance matrix might even have an unbounded spectrum. 

The results of the simulation study show that the dynamic shrinkage procedures derived in the paper are robust against violations of the model assumptions. In particular, we conclude based on the results of the simulation study, that the performance of the suggested dynamic approach will not be strongly influenced when the asset returns are generated from a multivariate GARCH model and from a VARMA model. Although, both multivariate times series model assume that the asset returns are time dependent, it has a minor influence the suggested trading strategies. Finally, we apply the new approaches to real data of returns on stocks included in the S\&P 500 index and compare them with several benchmark approaches. These consists of investing into the target portfolio and the sample GMV portfolio amongst others. Several performance measures are considered and it is shown that the dynamic shrinkage portfolio constructed by using non-overlapping samples possesses the best performance in terms of the turnover. On the other hand, the dynamic shrinkage portfolio constructed by using overlapping yield the best return.

The dynamic strategies based non-overlapping sample are simple to implement and they provide sometimes drastically less turnover in comparison to the benchmark approaches. Although the approaches based on the overlapping estimators are harder to implement, they sometimes increase the return by twice the amount of the non-overlapping strategies. Furthermore, the overlapping estimators only require that the sample size is larger than the portfolio dimension in the first period. The non-overlapping approaches demands that the sample size is larger than the portfolio dimension at each reconstruction of the portfolio.

No portfolio is ever static. Making optimal transitions are therefore of great interest to any investor. These results provide a fully data-driven dynamic approaches how the GMV portfolio can be rebalanced. In many practical applications the investors might want to have more assets in their portfolios than the available sample size. This demands a special attention since the sample covariance matrix is singular and its inverse does no longer exist. This challenging problem has not been treated in the paper and is left for future research.

\section*{Acknowledgement} 
This research was partly supported by the Swedish Research Council (VR) via the project ''Bayesian Analysis of Optimal Portfolios and Their Risk Measures''


{\footnotesize
\bibliography{references}
}

\section{Appendix: Proofs}

In this section the proofs of the theoretical results are given. We first state several lemmas which will be used in the proofs of the theorems.

For any integer $n>2$, we define
\begin{equation}
\bV_{n}=\frac{1}{n-1}\bX_{n}\left(\bI_{n}-\frac{1}{n}\bOne_{n}\bOne_{n}^\top\right)\bX_{n}^\top
\quad \text{and} \quad
\widetilde{\bV}_{n}=\frac{1}{n-1}\bX_{n}\bX_{n}^\top,
\end{equation}
where $\bX_{n}$ is given in \eqref{obs_i} for the special case $n=n_i$. Hence,
\[\bS_{n}=\bSigma^{1/2}\bV_{n}\bSigma^{1/2}
=\bSigma^{1/2}\widetilde{\bV}_{n}\bSigma^{1/2}
-\frac{n}{n-1}\bSigma^{1/2}\bar\bx_{n}\bar\bx_{n}^\top\bSigma^{1/2}\]
with $\bar\bx_{n}=\frac{1}{n}\bX_{n}\bOne_{n}
=\bSigma^{-1/2}\bar\by_n$.\\

The statement of Lemma \ref{lem:lem1} is derived as Lemma 5.3 in \citet{bodnarokhrinparolya2020}
\begin{lemma}\label{lem:lem1}
    Let $\bxi$ and $\btheta$ be two nonrandom vectors with bounded Euclidean norms. Then it holds that
    \begin{align}
        \left|\bxi^\top \bV_{n}^{-1} \btheta - (1-c)^{-1}\bxi^\top\btheta  \right| &\stackrel{a.s.}{\rightarrow} 0 \label{eqn:lem1eq1}\\
        \left|\bxi^\top \bV_{n}^{-2} \btheta - (1-c)^{-3}\bxi^\top\btheta  \right| &\stackrel{a.s.}{\rightarrow} 0 \label{eqn:lem1eq4}
    \end{align}
    for $p/n \rightarrow c \in (0, 1)$
    as $n\rightarrow \infty$.
\end{lemma}

\begin{lemma}\label{lem:lem2}
Let $\bxi$ and $\btheta$ be two nonrandom vectors with bounded Euclidean norms and let $m,n>1$. Then it holds that
    \begin{align}
        \left|\bxi^\top \widetilde{\bV}_{n}^{-1} \widetilde{\bV}_{n+m}^{-1}\btheta - d_{n,n+m}\bxi^\top\btheta  \right| &\stackrel{a.s.}{\rightarrow} 0, \label{eqn:lem2eq1}
    \end{align}
for $p/n \rightarrow c \in (0, 1)$
    as $n\rightarrow \infty$ with
\begin{equation} \label{eqn:lem2eq2}
d_{n,n+m}=\frac{b_{n,n+m}}{a_{n,n+m}}\left(\left(1-\frac{p}{n}\right)^{-1}
-2\left(1-\frac{p}{n}+a_{n,n+m}+\sqrt{\left(1-\frac{p}{n}-a_{n,n+m}\right)^2+4a_{n,n+m}}\right)^{-1}\right),
\end{equation}
where
\begin{equation} \label{eqn:lem2eq3}
a_{n,n+m}=\frac{n+m-p}{n+m} \frac{n+m-1}{n-1}
\quad \text{and} \quad
b_{n,n+m}=\frac{n+m-1}{n-1}.
\end{equation}
\end{lemma}

\begin{proof}[Proof of Lemma \ref{lem:lem2}]
It holds that
\begin{eqnarray}
&&\left|\bxi^\top \widetilde{\bV}_{n}^{-1} \widetilde{\bV}_{n+m}^{-1}\btheta -d_{n,n+m}\bxi^\top\btheta  \right|
\nonumber\\
&\le&\left|\frac{\bxi^\top \widetilde{\bV}_{n}^{-1} \widetilde{\bV}_{n+m}^{-1}\btheta}{\sqrt{\bxi^\top \widetilde{\bV}_{n}^{-2}\bxi}\sqrt{\btheta^\top\btheta}} - \frac{\bxi^\top\widetilde{\bV}_{n}^{-1}\left(\frac{n-1}{n+m-1}\widetilde{\bV}_{n}+\frac{n+m-p}{n+m}\bI\right)^{-1}\btheta}
{\sqrt{\bxi^\top \widetilde{\bV}_{n}^{-2}\bxi}\sqrt{\btheta^\top\btheta}} \right|
\sqrt{\bxi^\top \widetilde{\bV}_{n}^{-2}\bxi}\sqrt{\btheta^\top\btheta}\nonumber\\
&+&\left|\bxi^\top\widetilde{\bV}_{n}^{-1}\left(\frac{n-1}{n+m-1}\widetilde{\bV}_{n}+\frac{n+m-p}{n+m}\bI\right)^{-1}\btheta - d_{n,n+m}\bxi^\top\btheta  \right| ,\label{lem2_eq0_pr}
\end{eqnarray}
where $\btheta^\top\btheta<\infty$ by the assumption and by Lemma 5.2 of \citet{bodnarokhrinparolya2020} we get $\bxi^\top \widetilde{\bV}_{n}^{-2}\bxi\stackrel{a.s.}{\rightarrow}(1-c)^{-3}\bxi^\top\bxi$.

Let $\widetilde{\bV}_{n+1:n+m}=\frac{1}{m-1}\bX_{n+1:n+m}\bX_{n+1:n+m}^\top$ where $\bX_{n+1:n+m}$ stands for the submatrix of $\bX_{n+m}$ consisting of its last $m$ columns. Since $\bX_{n}$ and $\bX_{n+1:n+m}$ are independent, $\widetilde{\bV}_{n}$ and $\widetilde{\bV}_{n+1:n+m}$ are also independent. Moreover, we get
\[\widetilde{\bV}_{n+m}=\frac{n-1}{n+m-1}\widetilde{\bV}_{n}+\frac{m-1}{n+m-1}\widetilde{\bV}_{n+1:n+m}.\]

Following Theorem 1 in \citet{rubio2011spectral} conditionally on  $\widetilde{\bV}_{n}$, we get that
\[\left|\frac{\bxi^\top \widetilde{\bV}_{n}^{-1} \widetilde{\bV}_{n+m}^{-1}\btheta}{\sqrt{\bxi^\top \widetilde{\bV}_{n}^{-2}\bxi}\sqrt{\btheta^\top\btheta}} - \frac{\bxi^\top\widetilde{\bV}_{n}^{-1}\left(\frac{n-1}{n+m-1}\widetilde{\bV}_{n}+k_{n,n+m}\bI\right)^{-1}\btheta}
{\sqrt{\bxi^\top \widetilde{\bV}_{n}^{-2}\bxi}\sqrt{\btheta^\top\btheta}} \right| \stackrel{a.s.}{\rightarrow} 0\]
for $p/m \rightarrow \tilde c \in (0, \infty)$ as $m\rightarrow \infty$
where $k_{n,n+m}$ is found as the solution of the equation
\begin{eqnarray}\label{lem2_eq1_pr}
\frac{m}{p}(k_{n,n+m}^{-1}-1)&=&\frac{1}{p}\text{tr}\left(\frac{m}{n+m-1}\left(\frac{n-1}{n+m-1}\widetilde{\bV}_{n}+k_{n,n+m}\frac{m}{n+m-1}\bI\right)^{-1}\right)\\
 &\stackrel{a.s.}{\rightarrow}& 2\frac{m}{n}\left(1-\frac{p}{n}+k_{n,n+m}\frac{m}{n}+
 \sqrt{\left(1-\frac{p}{n}-k_{n,n+m}\frac{m}{n}\right)^2+4k_{n,n+m}\frac{m}{n}}\right)^{-1}\nonumber
\end{eqnarray}

Using that for $a,b\in \mathds{R}$ we have
\[\left((a+b)+\sqrt{(a-b)^2+4b}\right)\left((a+b)-\sqrt{(a-b)^2+4b}\right)
=4b(a-1),
\]
the identity \eqref{lem2_eq1_pr} is equivalent to
\begin{eqnarray*}
2\frac{m}{n}(k_{n,n+m}-1)&=&
1-\frac{p}{n}+k_{n,n+m}\frac{m}{n}-
 \sqrt{\left(1-\frac{p}{n}-k_{n,n+m}\frac{m}{n}\right)^2+4k_{n,n+m}\frac{m}{n}}
\end{eqnarray*}
or
\begin{eqnarray*}
\sqrt{\left(1-\frac{p}{n}-k_{n,n+m}\frac{m}{n}\right)^2+4k_{n,n+m}\frac{m}{n}}&=&
1-\frac{p}{n}+2\frac{m}{n}-\frac{m}{n}k_{n,n+m},
\end{eqnarray*}
whose solution is given by
\[k_{n,n+m}=\frac{n+m-p}{n+m}.\]

For the second summand in \eqref{lem2_eq0_pr} we get that
\begin{eqnarray*}
&&\bxi^\top\widetilde{\bV}_{n}^{-1}\left(\frac{n-1}{n+m-1}\widetilde{\bV}_{n}+\frac{n+m-p}{n+m}\bI\right)^{-1}\btheta
=b_{n,n+m}\bxi^\top\widetilde{\bV}_{n}^{-1}\left(\widetilde{\bV}_{n}+a_{n,n+m}\bI\right)^{-1}\btheta\\
&=&\frac{b_{n,n+m}}{a_{n,n+m}}\left(\bxi^\top\widetilde{\bV}_{n}^{-1}\btheta
-\bxi^\top\left(\widetilde{\bV}_{n}+a_{n,n+m}\bI\right)^{-1}\btheta\right)\\
&\stackrel{a.s.}{\rightarrow}& \frac{b_{n,n+m}}{a_{n,n+m}}\left(\left(1-\frac{p}{n}\right)^{-1}
-2\left(1-\frac{p}{n}+a_{n,n+m}+\sqrt{\left(1-\frac{p}{n}-a_{n,n+m}\right)^2+4a_{n,n+m}}\right)^{-1}\right)\bxi^\top\btheta
\end{eqnarray*}
using Lemma 5.1 in \citet{bodnarokhrinparolya2020} two times.
\end{proof}

\begin{lemma}\label{lem:lem3}
Let $\bxi$ and $\btheta$ be two nonrandom vectors with bounded Euclidean norms and let $m,n>1$. Then it holds that
    \begin{align}
        \left|\bxi^\top \bV_{n}^{-1} \bV_{n+m}^{-1}\btheta - d_{n,n+m}\bxi^\top\btheta  \right| &\stackrel{a.s.}{\rightarrow} 0, \label{eqn:lem3eq1}
    \end{align}
for $p/n \rightarrow c \in (0, 1)$
    as $n\rightarrow \infty$ where $d_{n,n+m}$ is defined in \eqref{eqn:lem2eq2}.
\end{lemma}

\begin{proof}[Proof of Lemma \ref{lem:lem3}]
The application of the Sherman-Morrison formula leads to
\begin{eqnarray*}
&&\bxi^\top \bV_{n}^{-1} \bV_{n+m}^{-1}\btheta=
\bxi^\top \left(\widetilde{\bV}_{n}-\frac{n}{n-1}\bar{\bx}_n\bar{\bx}_n^\top\right)^{-1} \left(\widetilde{\bV}_{n+m}-\frac{n+m}{n+m-1}\bar{\bx}_{n+m}\bar{\bx}_{n+m}^\top\right)^{-1}\btheta
\\
&=& \bxi^\top \widetilde{\bV}_{n}^{-1} \widetilde{\bV}_{n+m}^{-1}\btheta+\frac{n+m}{n+m-1}\frac{
\bxi^\top \widetilde{\bV}_{n}^{-1} \widetilde{\bV}_{n+m}^{-1}\bar{\bx}_{n+m}\bar{\bx}_{n+m}^\top\widetilde{\bV}_{n+m}^{-1}\btheta}{1-\frac{n+m}{n+m-1}\bar{\bx}_{n+m}^\top\widetilde{\bV}_{n+m}^{-1}\bar{\bx}_{n+m}}\\
&+&\frac{n}{n-1}\frac{
\bxi^\top \widetilde{\bV}_{n}^{-1}\bar{\bx}_{n}\bar{\bx}_{n}^\top \widetilde{\bV}_{n}^{-1}\widetilde{\bV}_{n+m}^{-1}\btheta}{1-\frac{n}{n-1}\bar{\bx}_{n}^\top\widetilde{\bV}_{n}^{-1}\bar{\bx}_{n}}\\
&+&\frac{n}{n-1}\frac{n+m}{n+m-1}
\frac{
\bxi^\top \widetilde{\bV}_{n}^{-1}\bar{\bx}_{n}\bar{\bx}_{n}^\top \widetilde{\bV}_{n}^{-1} \widetilde{\bV}_{n+m}^{-1}\bar{\bx}_{n+m}\bar{\bx}_{n+m}^\top\widetilde{\bV}_{n+m}^{-1}\btheta}{(1-\frac{n}{n-1}\bar{\bx}_{n}^\top\widetilde{\bV}_{n}^{-1}\bar{\bx}_{n})(1-\frac{n+m}{n+m-1}\bar{\bx}_{n+m}^\top\widetilde{\bV}_{n+m}^{-1}\bar{\bx}_{n+m})}.
\end{eqnarray*}
The applications of Lemma 5.2 in \citet{bodnarokhrinparolya2020} and Lemma \ref{lem:lem2} completes the proof.
\end{proof}

\begin{proof}[Proof of Theorem \ref{th1}]
Rewriting \eqref{alp_ni_star} we get
\begin{eqnarray*}
    \psi_{n_i}^*&=& \frac{\hbw_{SH;n_{i-1}}^\top \bSigma\hbw_{SH;n_{i-1}}-\frac{\hbw_{SH;n_{i-1}}^\top \bSigma \bS_{n_i}^{-1}\bOne_p}{\bOne_p^\top\bS_{n_i}^{-1}\bOne_p}} {\hbw_{SH;n_{i-1}}^\top \bSigma\hbw_{SH;n_{i-1}}
    -2\frac{\hbw_{SH;n_{i-1}}^\top \bSigma \bS_{n_i}^{-1}\bOne_p}{\bOne_p^\top\bS_{n_i}^{-1}\bOne_p}+\frac{\bOne_p^\top\bS_{n_i}^{-1} \bSigma \bS_{n_i}^{-1}\bOne_p}{(\bOne_p^\top\bS_{n_i}^{-1}\bOne_p)^2}}=\frac{h_{n_i}}{g_{n_i}}
\end{eqnarray*}
with
\begin{eqnarray*}
 h_{n_i} &=&\bOne_p^\top \bSigma^{-1}\bOne_p\hbw_{SH;n_{i-1}}^\top \bSigma\hbw_{SH;n_{i-1}}
 -\frac{\bOne_p^\top \bSigma^{-1}\bOne_p}{\bOne_p^\top\bS_{n_i}^{-1}\bOne_p}\hbw_{SH;n_{i-1}}^\top \bSigma \bS_{n_i}^{-1}\bOne_p
\end{eqnarray*}
and
\begin{eqnarray*}
 g_{n_i} &=&
  {\bOne_p^\top \bSigma^{-1}\bOne_p}\hbw_{SH;n_{i-1}}^\top \bSigma\hbw_{SH;n_{i-1}}+\frac{\bOne_p^\top\bS_{n_i}^{-1} \bSigma \bS_{n_i}^{-1}\bOne_p}{\bOne_p^\top \bSigma^{-1}\bOne_p}
    \left(\frac{\bOne_p^\top \bSigma^{-1}\bOne_p}{\bOne_p^\top\bS_{n_i}^{-1}\bOne_p}\right)^2\\
    &-&2\frac{\bOne_p^\top \bSigma^{-1}\bOne_p}{\bOne_p^\top\bS_{n_i}^{-1}\bOne_p}
    \hbw_{SH;n_{i-1}}^\top \bSigma \bS_{n_i}^{-1}\bOne_p
\end{eqnarray*}

The application of Lemma \ref{lem:lem1} with $\bxi=\btheta=\dfrac{\bSigma^{-1/2}\bOne_p}{\bOne_p^\top \bSigma^{-1}\bOne_p}$ leads to
\begin{equation}\label{th1_proof_eq1}
\frac{\bOne_p^\top\bS_{n_i}^{-1}\bOne_p }{\bOne_p^\top \bSigma^{-1}\bOne_p}   \stackrel{a.s.}{\rightarrow} (1-c_i)^{-1} 
\quad
\text{and}
\quad
\frac{\bOne_p^\top\bS_{n_i}^{-1} \bSigma \bS_{n_i}^{-1}\bOne_p}{\bOne_p^\top \bSigma^{-1}\bOne_p}  \stackrel{a.s.}{\rightarrow} (1-c_i)^{-3} 
\end{equation}
for $p/n_i \rightarrow c_i \in (0, 1)$ as $n_i\rightarrow \infty$.

Since the variance of the GMV portfolio is equal to $1/(\bOne_p^\top \bSigma^{-1}\bOne_p)$, we have that
\[\bOne_p^\top \bSigma^{-1}\bOne_p
    \hbw_{SH;n_{i-1}}^\top \bSigma\hbw_{SH;n_{i-1}}=\frac{\hbw_{SH;n_{i-1}}^\top \bSigma\hbw_{SH;n_{i-1}}}{1/\bOne_p^\top \bSigma^{-1}\bOne_p}=1+r_{\hbw_{SH;n_{i-1}}}\ge1,\]
where $r_{\hbw_{SH;n_{i-1}}}$ is the relative loss of the portfolio with weights $\hbw_{SH;n_{i-1}}$ computed with respect to the variance of the GMV portfolio. Next, we recursively derive the asymptotic behaviour of $r_{\hbw_{SH;n_{i-1}}}$. 

For $i=1$ the sample estimator of the GMV portfolio weights $\hbw_{S;n_i}$ is shrunk to the deterministic target $\bb$, i.e., $\hbw_{SH;n_{0}}=\bb$. In this case the relative loss is given by
\begin{eqnarray*}
r_{\hbw_{SH;n_{0}}}=r_{0}=\bOne_p^\top \bSigma^{-1}\bOne_p\bb^\top\bSigma \bb-1,
\end{eqnarray*}
which is bounded uniformly on $p$ following the assumption of the theorem and
\[\frac{\bb^\top \bSigma \bS_{n_1}^{-1}\bOne_p}{\sqrt{\bb^\top \bSigma\bb}\sqrt{\bOne_p^\top \bSigma^{-1}\bOne_p}}\stackrel{a.s.}{\rightarrow} (1-c_1)^{-1} \frac{1}{\sqrt{\bb^\top \bSigma\bb}\sqrt{\bOne_p^\top \bSigma^{-1}\bOne_p}}=(1-c_1)^{-1}\frac{1}{\sqrt{r_{0}+1}},\]
using Lemma \ref{lem:lem1} with $\bxi=\dfrac{\bSigma^{1/2}\bb}{\sqrt{\bb^\top \bSigma\bb}}$ and $\btheta=\dfrac{\bSigma^{-1/2}\bOne_p}{\bOne_p^\top \bSigma^{-1}\bOne_p}$. Combining these results with \eqref{th1_proof_eq1} leads to
\[\psi_{n_1}^* \stackrel{a.s.}{\rightarrow} \psi_1^* \]
for $p/n_1 \rightarrow c_1 \in (0, 1)$ as $n_1\rightarrow \infty$ with
\[\psi_1^*=\frac{(r_{0}+1)-1}
{(r_{0}+1)+(1-c_1)^{-1}-2}
=\frac{(1-c_1)r_{0}}{(1-c_1)r_{0}+c_1}.
\]
Moreover, the relative loss of the portfolio $\hbw_{SH;n_{i-1}}$ tends to 
\begin{eqnarray*}
r_{\hbw_{SH;n_{1}}}&=& \bOne_p^\top \bSigma^{-1}\bOne_p\left((\psi_{n_1}^*)^2\hbw_{S;n_1}^\top \bSigma \hbw_{S;n_1}+ 2\psi_{n_1}^*(1-\psi_{n_1}^*)\hbw_{S;n_1}^\top\bSigma \bb+ (1-\psi_{n_1}^*)^2\bb^\top\bSigma \bb\right)-1\\
&=& (\psi_{n_1}^*)^2\frac{\bOne_p^\top\bS_{n_1}^{-1} \bSigma \bS_{n_1}^{-1}\bOne_p}{\bOne_p^\top \bSigma^{-1}\bOne_p}\left(\frac{\bOne_p^\top \bSigma^{-1}\bOne_p}{\bOne_p^\top\bS_{n_1}^{-1}\bOne_p}\right)^2\\
&+& 2\psi_{n_1}^*(1-\psi_{n_1}^*)\bOne_p^\top\bS_{n_1}^{-1} \bSigma\bb
\frac{\bOne_p^\top \bSigma^{-1}\bOne_p}{\bOne_p^\top\bS_{n_1}^{-1}\bOne_p}
+ (1-\psi_{n_1}^*)^2(r_{0}+1)-1\\
&\stackrel{a.s.}{\rightarrow}&
 (\psi_{1}^*)^2(1-c_1)^{-1}+2\psi_{1}^*(1-\psi_{1}^*)+(1-\psi_{1}^*)^2(r_{0}+1)-1\\
&=&(\psi_{1}^*)^2\frac{c_1}{1-c_1}+(1-\psi_{1}^*)^2r_{0}=r_1,
\end{eqnarray*}
for $p/n_1 \rightarrow c_1 \in (0, 1)$ as $n_1\rightarrow \infty$.

Using the last result we get for $i=2$ that
\[\bOne_p^\top \bSigma^{-1}\bOne_p
    \hbw_{SH;n_{1}}^\top \bSigma\hbw_{SH;n_{1}}\stackrel{a.s.}{\rightarrow}
r_1+1.
\]
Since non-overlaping samples $\bY_{n_1}$ $\bY_{n_2}$ are used in the computation of the $\hbw_{SH;n_{1}}$ and $\bS_{n_2}$, these two random objects are independent. The application of Lemma \ref{lem:lem1} conditionally on $\bY_{n_1}$ leads to
\begin{eqnarray*}
\left|\frac{\hbw_{SH;n_{1}}^\top \bSigma \bS_{n_2}^{-1}\bOne_p}{\sqrt{\hbw_{SH;n_{1}}^\top \bSigma\hbw_{SH;n_{1}}}\sqrt{\bOne_p^\top \bSigma^{-1}\bOne_p}}-
\frac{(1-c_2)^{-1}}{\sqrt{\hbw_{SH;n_{1}}^\top \bSigma\hbw_{SH;n_{1}}}\sqrt{\bOne_p^\top \bSigma^{-1}\bOne_p}}
\right|\stackrel{a.s.}{\rightarrow} 0,
\end{eqnarray*}
and, hence,
\begin{eqnarray*}
&&\left|\hbw_{SH;n_{1}}^\top \bSigma \bS_{n_2}^{-1}\bOne_p-(1-c_2)^{-1}\right|\\
&=&\left|\frac{\hbw_{SH;n_{1}}^\top \bSigma \bS_{n_2}^{-1}\bOne_p}{\sqrt{\hbw_{SH;n_{1}}^\top \bSigma\hbw_{SH;n_{1}}}\sqrt{\bOne_p^\top \bSigma^{-1}\bOne_p}}\sqrt{\hbw_{SH;n_{1}}^\top \bSigma\hbw_{SH;n_{1}}}\sqrt{\bOne_p^\top \bSigma^{-1}\bOne_p}-
(1-c_2)^{-1}\right|\stackrel{a.s.}{\rightarrow} 0,
\end{eqnarray*}
since $\hbw_{SH;n_{1}}^\top \bSigma\hbw_{SH;n_{1}}\bOne_p^\top \bSigma^{-1}\bOne_p=O_P(1)$
The last results together with \eqref{th1_proof_eq1} yields
\[\psi_{n_2}^* \stackrel{a.s.}{\rightarrow} \psi_2^*= \frac{(r_{1}+1)-1}
{(r_{1}+1)+(1-c_2)^{-1}-2}
=\frac{(1-c_2)r_{1}}{(1-c_2)r_{1}+c_2}
\]
for $p/n_2 \rightarrow c_2 \in (0, 1)$ as $n_2\rightarrow \infty$.

Finally, the relative loss of the portfolio $\hbw_{SH;n_{2}}$ tends to 
\begin{eqnarray*}
r_{\hbw_{SH;n_{2}}}&=& \bOne_p^\top \bSigma^{-1}\bOne_p\Bigg((\psi_{n_2}^*)^2\hbw_{S;n_2}^\top \bSigma \hbw_{S;n_2}+ 2\psi_{n_2}^*(1-\psi_{n_2}^*)\hbw_{S;n_2}^\top\bSigma \hbw_{SH;n_{1}}\\
&+& (1-\psi_{n_2}^*)^2\hbw_{SH;n_{1}}^\top\bSigma \hbw_{SH;n_{1}}\Bigg)-1\\
&=& (\psi_{n_2}^*)^2\frac{\bOne_p^\top\bS_{n_2}^{-1} \bSigma \bS_{n_2}^{-1}\bOne_p}{\bOne_p^\top \bSigma^{-1}\bOne_p}\left(\frac{\bOne_p^\top \bSigma^{-1}\bOne_p}{\bOne_p^\top\bS_{n_2}^{-1}\bOne_p}\right)^2\\
&+& 2\psi_{n_2}^*(1-\psi_{n_2}^*)\bOne_p^\top\bS_{n_2}^{-1} \bSigma\hbw_{SH;n_{1}}
\frac{\bOne_p^\top \bSigma^{-1}\bOne_p}{\bOne_p^\top\bS_{n_2}^{-1}\bOne_p}
+ (1-\psi_{n_2}^*)^2(r_{\hbw_{SH;n_{1}}}+1)-1\\
&\stackrel{a.s.}{\rightarrow}&
 (\psi_{2}^*)^2(1-c_2)^{-1}+2\psi_{2}^*(1-\psi_{2}^*)+(1-\psi_{2}^*)^2(r_{1}+1)-1\\
&=&(\psi_{2}^*)^2\frac{c_2}{1-c_2}+(1-\psi_{2}^*)^2r_{1}=r_2,
\end{eqnarray*}
for $p/n_2 \rightarrow c_2 \in (0, 1)$ as $n_2\rightarrow \infty$.

Repeating the above steps for $i=3,...,T$ leads to the statement of the theorem.
 \end{proof}
\begin{proof}[Proof of Theorem \ref{th2}]
We get
\begin{equation*}
   \Psi_{N_i}^*= \frac{\hbw_{SH;N_{i-1}}^\top \bSigma\hbw_{SH;N_{i-1}}-\frac{\hbw_{SH;N_{i-1}}^\top \bSigma \bS_{N_i}^{-1}\bOne_p}{\bOne_p^\top\bS_{N_i}^{-1}\bOne_p}} {\hbw_{SH;N_{i-1}}^\top \bSigma\hbw_{SH;N_{i-1}}
    -2\frac{\hbw_{SH;N_{i-1}}^\top \bSigma \bS_{n_i}^{-1}\bOne_p}{\bOne_p^\top\bS_{N_i}^{-1}\bOne_p}+\frac{\bOne_p^\top\bS_{N_i}^{-1} \bSigma \bS_{N_i}^{-1}\bOne_p}{(\bOne_p^\top\bS_{N_i}^{-1}\bOne_p)^2}}
   =\frac{H_{N_i}}{G_{N_i}}
\end{equation*}
with
\begin{eqnarray*}
 H_{N_i} &=&\bOne_p^\top \bSigma^{-1}\bOne_p\hbw_{SH;N_{i-1}}^\top \bSigma\hbw_{SH;N_{i-1}}
 -\frac{\bOne_p^\top \bSigma^{-1}\bOne_p}{\bOne_p^\top\bS_{N_i}^{-1}\bOne_p}\hbw_{SH;N_{i-1}}^\top \bSigma \bS_{N_i}^{-1}\bOne_p
\end{eqnarray*}
and
\begin{eqnarray*}
 G_{N_i} &=&
  {\bOne_p^\top \bSigma^{-1}\bOne_p}\hbw_{SH;N_{i-1}}^\top \bSigma\hbw_{SH;N_{i-1}}+\frac{\bOne_p^\top\bS_{N_i}^{-1} \bSigma \bS_{N_i}^{-1}\bOne_p}{\bOne_p^\top \bSigma^{-1}\bOne_p}
    \left(\frac{\bOne_p^\top \bSigma^{-1}\bOne_p}{\bOne_p^\top\bS_{N_i}^{-1}\bOne_p}\right)^2\\
    &-&2\frac{\bOne_p^\top \bSigma^{-1}\bOne_p}{\bOne_p^\top\bS_{N_i}^{-1}\bOne_p}
    \hbw_{SH;N_{i-1}}^\top \bSigma \bS_{N_i}^{-1}\bOne_p
 \end{eqnarray*}

From Lemma \ref{lem:lem1} we have 
\begin{equation}\label{th2_proof_eq1}
\frac{\bOne_p^\top\bS_{N_i}^{-1}\bOne_p }{\bOne_p^\top \bSigma^{-1}\bOne_p}   \stackrel{a.s.}{\rightarrow} (1-C_i)^{-1} 
\quad
\text{and}
\quad
\frac{\bOne_p^\top\bS_{N_i}^{-1} \bSigma \bS_{N_i}^{-1}\bOne_p}{\bOne_p^\top \bSigma^{-1}\bOne_p}  \stackrel{a.s.}{\rightarrow} (1-C_i)^{-3} 
\end{equation}
for $p/N_i \rightarrow C_i \in (0, 1)$ as $N_i\rightarrow \infty$.

The recursive structure of $\hbw_{SH;N_{i-1}}$ implies that
\begin{eqnarray*}
\hbw_{SH;N_{i-1}}=\sum_{j=0}^{i-1} \beta^{*}_{N_{i-1};N_j}\hbw_{N_{j}}
\quad \text{with} \quad \hbw_{N_{0}}=\bb,
\end{eqnarray*}
where $\beta^{*}_{N_{i-1};N_j}$ are computed recursively by
\begin{equation*}
\beta^{*}_{N_0;N_0}=1, \quad \beta^{*}_{N_{i-1};N_{i-1}}= \Psi_{N_{i-1}}^*,
\quad \text{and} \quad \beta^{*}_{N_{i-1};N_j}= (1-\Psi_{N_{i-1}}^*)\beta^{*}_{N_{i-2};N_j}, \, j=0,...i-2.
\end{equation*}
Hence,
\begin{eqnarray*}
\hbw_{SH;N_{i-1}}^\top \bSigma \bS_{N_i}^{-1}\bOne_p
&=&\beta^{*}_{N_{i-1};N_0} \bb^\top\bSigma \bS_{N_i}^{-1}\bOne_p+
\sum_{j=1}^{i-1} \beta^{*}_{N_{i-1};N_j}\frac{\bOne_p^\top\bS_{N_j}^{-1}\bSigma \bS_{N_i}^{-1}\bOne_p}{\bOne_p^\top\bS_{N_j}^{-1}\bOne_p}
\end{eqnarray*}

Let $\Psi_{i-1}^*$ denote the deterministic asymptotic limit of $\Psi_{N_{i-1}}^*$, whose recursive computation is discussed below. We also define
\begin{equation*}
\beta^{*}_{0;0}=1, \quad \beta^{*}_{i-1;i-1}= \Psi_{i-1}^*,
\quad \text{and} \quad \beta^{*}_{i-1;j}= (1-\Psi_{i-1}^*)\beta^{*}_{i-2;j}, \, j=0,...i-2.
\end{equation*}
Then, the application of Lemma \ref{lem:lem3} yields
\[
\hbw_{SH;N_{i-1}}^\top \bSigma \bS_{N_i}^{-1}\bOne_p\stackrel{a.s.}{\rightarrow}
(1-C_i)^{-1}K_i
\quad \text{with} \quad K_i=\beta^{*}_{i-1;0}+\sum_{j=1}^{i-1} \beta^{*}_{i-1;j}D_{j,i}
\]
for $p/N_j \rightarrow C_j \in (0, 1)$ as $N_j\rightarrow \infty$, $j=1,...,i-1$, where
\[D_{j,i}=1-\frac{2(1-C_j)}{(1-C_j)+(1-C_i)\frac{C_j}{C_i}+\sqrt{\left(1-\frac{C_j}{C_i}\right)^2+4(1-C_i)\frac{C_j}{C_i}}}.
\]

Moreover, the relative loss of the portfolio with weights $\hbw_{SH;N_{i}}$ is asymptotically given by
\begin{eqnarray*}
R_{\hbw_{SH;N_{i}}}&=& \bOne_p^\top \bSigma^{-1}\bOne_p\Bigg((\Psi_{N_i}^*)^2\hbw_{S;n_i}^\top \bSigma \hbw_{S;N_i}+ 2\Psi_{N_i}^*(1-\psi_{N_i}^*)\hbw_{S;N_i}^\top\bSigma \hbw_{SH;N_{i-1}}\\
&+& (1-\Psi_{N_i}^*)^2\hbw_{SH;N_{i-1}}^\top\bSigma \hbw_{SH;N_{i-1}}\Bigg)-1\\
&=& (\Psi_{N_i}^*)^2\frac{\bOne_p^\top\bS_{N_i}^{-1} \bSigma \bS_{N_i}^{-1}\bOne_p}{\bOne_p^\top \bSigma^{-1}\bOne_p}\left(\frac{\bOne_p^\top \bSigma^{-1}\bOne_p}{\bOne_p^\top\bS_{N_i}^{-1}\bOne_p}\right)^2\\
&+& 2\Psi_{N_i}^*(1-\Psi_{N_i}^*)\bOne_p^\top\bS_{N_i}^{-1} \bSigma\hbw_{SH;N_{i-1}}
\frac{\bOne_p^\top \bSigma^{-1}\bOne_p}{\bOne_p^\top\bS_{N_{i}}^{-1}\bOne_p}
+ (1-\Psi_{N_i}^*)^2(r_{\hbw_{SH;N_{i-1}}}+1)-1\\
&\stackrel{a.s.}{\rightarrow}&R_i=
  (\Psi_{i}^*)^2(1-C_i)^{-1}+2\Psi_{i}^*(1-\Psi_{i}^*)K_i+(1-\Psi_{i}^*)^2(R_{i-1}+1)-1\\
&=&(\Psi_{i}^*)^2\frac{C_i}{1-C_i}+(1-\Psi_{i}^*)^2R_{i-1}+2\Psi_{i}^*(1-\Psi_{i}^*)(K_i-1).
\end{eqnarray*}

Finally, we get
\[\Psi_{N_i}^* \stackrel{a.s.}{\rightarrow} \Psi_i^* 
\quad \text{with} \quad
\Psi_i^*=\frac{(R_{i-1}+1)-K_i}
{(R_{i-1}+1)+(1-C_i)^{-1}-2K_i},
\]
for $p/N_j \rightarrow C_j \in (0, 1)$ as $N_j\rightarrow \infty$, $j=1,...,i$. As a result, the computation of $\Psi_i^*$ is performed in the following recursive way: (i) first, we compute $R_0=\bOne_p^\top \bSigma^{-1}\bOne_p\bb^\top\bSigma \bb-1$ and $K_1=(1-C_1)^{-1}$ which are used to obtain $\Psi_1^*$; (ii) second, using $R_0$ and $\Psi_1^*$ we find $R_1$ and $K_1=(1-C_1)^2$ used in the computation of $\Psi_2^*$ and proceed the recursive procedure for $i=3,...,T$. The boundedness of $R_0$ ensures that all computed values are finite as well.
\end{proof}

\end{document}